\documentclass[a4paper,10pt]{amsart}

\RequirePackage[OT1]{fontenc}
\RequirePackage{amsthm,amsmath}
\RequirePackage[numbers]{natbib}
\RequirePackage[colorlinks,citecolor=blue,urlcolor=blue]{hyperref}
\usepackage[dvipdfmx]{graphicx}

\usepackage[margin=2.5cm]{geometry}
\usepackage{amssymb}
\usepackage{ascmac}
\usepackage{xcolor}
\usepackage{bm,bbm}
\usepackage{fancyvrb}

\usepackage{tikz}
\usetikzlibrary{positioning,shapes,shadows,arrows}

\newtheorem{Def}{Definition}[section]
\newtheorem{Thm}[Def]{Theorem}
\newtheorem{Lem}[Def]{Lemma}
\newtheorem{Assumption}[Def]{Assumption}
\newtheorem{Rem}[Def]{Remark}
\newtheorem{Cor}[Def]{Corollary}

\newtheorem{Example}[Def]{Example}

\allowdisplaybreaks

\newcommand{\mca}{\mathcal{A}}

\newcommand{\mcc}{\mathcal{C}}

\newcommand{\mcf}{\mathcal{F}}

\newcommand{\mci}{\mathcal{I}}

\newcommand{\mcl}{\mathcal{L}}

\newcommand{\mbbh}{\mathbb{H}}
\newcommand{\mbbn}{\mathbb{N}}

\newcommand{\mbbr}{\mathbb{R}}

\newcommand{\mbby}{\mathbb{Y}}

\newcommand{\al}{\alpha}
\newcommand{\del}{\delta}
\newcommand{\ep}{\epsilon} 
\newcommand{\vp}{\varphi}

\newcommand{\D}{\Delta}
\newcommand{\sig}{\sigma}
\newcommand{\Sig}{\Sigma}
\newcommand{\lam}{\lambda}

\newcommand{\gam}{\gamma}
\newcommand{\Gam}{\Gamma}
\newcommand{\p}{\partial}

\newcommand{\cil}{\xrightarrow{\mcl}} 
\newcommand{\cip}{\xrightarrow{p}} 
\newcommand{\argmax}{\mathop{\rm argmax}}
\newcommand{\diag}{\mathop{\rm diag}}

\newcommand{\lf}{\lfloor}
\newcommand{\rf}{\rfloor}

\def\ds#1{\displaystyle{#1}}
\def\nn{\nonumber}

\def\lp{L\'{e}vy process}
\def\lm{L\'{e}vy measure}

\def\sumj{\sum_{j=1}^{n}}

\def\pr{P}
\def\E{E}

\def\intj{\int_{\left\lf \frac{j-1}{h }\right\rf h }^{\left\lf \frac{j}{h }\right\rf h }}

\def\intj{\int_{t_{j-1}}^{t_j}}

\def\tz{\theta_{0}}

\def\tes{\hat{\theta}_{n}}

\def\ees{\hat{\eta}_n}
\def\aes{\hat{\alpha}_{n}}

\def\ges{\hat{\gamma}_{n}}

\newcommand{\yuima}{YUIMA }

\DefineVerbatimEnvironment{Sinput}{Verbatim}{fontshape=sl}
\DefineVerbatimEnvironment{Soutput}{Verbatim}{}
\DefineVerbatimEnvironment{Scode}{Verbatim}{fontshape=sl}

\DefineVerbatimEnvironment{Code}{Verbatim}{}
\DefineVerbatimEnvironment{CodeInput}{Verbatim}{fontshape=sl}
\DefineVerbatimEnvironment{CodeOutput}{Verbatim}{}
\newenvironment{CodeChunk}{}{}
\setkeys{Gin}{width=0.8\textwidth}

\let\proglang=\textsf
\let\code=\texttt

\title[Noise inference for ergodic L\'{e}vy driven SDE \footnote{This work was partly supported by 
JST CREST Grant Number JPMJCR14D7, Japan.}
]
{Noise inference for ergodic L\'{e}vy driven SDE}

\author[H. Masuda]{Hiroki Masuda}
\address{Department of Mathematical Sciences, Faculty of Mathematics, Kyushu University}
\email{hiroki@math.kyushu-u.ac.jp}

\author[L. Mercuri]{Lorenzo Mercuri}
\address{Department of Economics, Management and Quantitative Methods, University of Milan}
\email{lorenzo.mercuri@unimi.it}

\author[Y. Uehara]{Yuma Uehara}
\address{Department of Mathematics, Faculty of Engineering Science, Kansai University}
\email{y-uehara@kansai-u.ac.jp}

\begin{document}

\setlength{\baselineskip}{4.5mm}

\maketitle

\begin{abstract}
We study inference for the driving L\'{e}vy noise of an ergodic stochastic differential equation (SDE) model, when the process is observed at high-frequency and long time and when the drift and scale coefficients contain finite-dimensional unknown parameters. By making use of the Gaussian quasi-likelihood function for the coefficients, we derive a stochastic expansion for functionals of the unit-time residuals, which clarifies some quantitative effect of plugging in the estimators of the coefficients, thereby enabling us to take several inference procedures for the driving-noise characteristics into account.
We also present new classes and methods available in \yuima  for the simulation and the estimation of a L\'evy SDE model. We highlight the flexibility of these new advances in \yuima using simulated and real data.
\end{abstract}





\section{Introduction}

We consider the following univariate Markovian stochastic differential equation (SDE):
\begin{equation}
dX_t=a(X_t,\al)dt+c(X_{t-},\gam)dJ^\eta_t,
\label{hm:sde}
\end{equation}
where:
\begin{itemize}
\item The coefficients $a$ and $c$ are smooth enough with $c$ being non-degenerate, and known except for an unknown parameter
\begin{equation}
\theta:=(\gam,\al)\in\Theta_\gam\times\Theta_\al =: \Theta \subset\mbbr^p
\nonumber
\end{equation}
for bounded convex domains $\Theta_\al\subset\mbbr^{p_\al}$ and $\Theta_\gam\subset\mbbr^{p_\gam}$ ($p=p_\al+p_\gam$);
\item The driving noise $J^\eta$ is a standardized non-Gaussian L\'{e}vy process with finite moments, whose distribution depends on an unknown parameter $\eta\in\Theta_\eta$, a domain in $\mbbr^{p_\eta}$, and where $J^\eta$ is independent of $X_0$.
\end{itemize}
We write $J=J^\eta$ in the sequel. We suppose that there are true values $\tz=(\gam_0,\al_0)\in\Theta_\gam\times\Theta_\al$ and $\eta_0\in\Theta_\eta$ which induce the true image measure $\pr$ of $(X,J)$, and that we observe a discrete-time sample $\bm{X}_n:=(X_{t_j})_{j=1}^n$, where $t_j=t_j^n:=jh_n$, with the sampling stepsize $h=h_n\to 0$ satisfying that
\begin{equation}
T_n:=nh\to\infty \ \text{ and } nh^2\to 0  \text{ as }  n\to\infty,
\label{sampling.design}
\end{equation}
the so-called rapidly increasing experimental design.

Our objective is to estimate the value $\xi_0:=(\tz,\eta_0)$ under the ergodicity.
We remark the parameter may not completely characterize the distribution $\mcl(J)$, so that the problem is not necessarily parametric; for example, $\eta$ is just a skewness or kurtosis, which may or may not completely determine $\mcl(J_1)$.

For the estimation of the coefficient parameter $\theta$, from the statistical point of view it is important what kind of distributional is used to approximate the conditional distribution of $X_{t_j}$ given $X_{t_{j-1}}$.
Previously, \cite{Mas13-1} and \cite{MasUeh17-2} considered estimation of $\theta$ based on the Gaussian quasi-likelihood (GQL), and proved the asymptotic normality and the tail probability estimates of the Gaussian quasi-maximum likelihood estimator (GQMLE).
At the expense of efficiency, the GQL based method has the robustness against misspecification of the driving-noise distribution, which may be crucial in the context of time-series models, see \cite[Section 6.2]{Str05}.
Further, concerned with estimation of a L\'{e}vy-measure functional of the form $\int \vp(z)\nu(dz)$, with $\nu$ denoting the {\lm} of $J$, the previous study \cite{MasUeh17-1} proposed a moment-matching based estimator and proved its asymptotic normality at rate $\sqrt{T_n}$.
However, the procedure imposed some stringent conditions on the behavior of $\vp$ around the origin \cite[Assumption 2.7]{MasUeh17-1}, and was not quite suitable if we want to estimate $\mcl(J_1)$ directly; in general, estimation of the distribution and that of the corresponding {\lm} can be of technically rather different matters.

In this paper, we will propose yet another strategy for estimating $\eta$ based on the unit-time approximation, which goes as follows:

\begin{enumerate}
\item First we construct the GQMLE $\tes:=(\ges,\aes)$ and the residual
\begin{equation}
\label{hm:add1}
\widehat{\D_j J} = \widehat{\D_j J}^n := \frac{X_{t_j}-X_{t_{j-1}}-h a(X_{t_{j-1}},\aes)}{c(X_{t_{j-1}},\ges)}.
\end{equation}
\item We then estimate the $\mcl(J_1)$-i.i.d. sequence ($i=1,\dots,n$)
\begin{equation}
\ep_i:=J_{i}-J_{i-1}
\nonumber
\end{equation}
by adding up the finer increments over the $i$th unit-time interval $[i-1,i]$:
\begin{equation}
\label{hm:add2}
\hat{\ep}_i := 
\sum_{j\in A_i} \widehat{\D_j J},
\end{equation}
where, for each $i\in\{1,2,\dots,\lf T_n\rf\}$,
\begin{align}
A_i := \left\{j\in\mbbn:\, i-1 < t_j\le i\right\}= \left\{j\in\mbbn:\, \left\lf \frac{i-1}{h }\right\rf+1\leq j\leq\left\lf \frac{i}{h }\right\rf\right\},
\nonumber
\end{align}
and then measure the discrepancy between a functional of $\{\hat{\ep}_i\}$ and that of $\{\ep_i\}$ through a stochastic expansion.
\item For an appropriate function $m$, we estimate $\eta$ by
\begin{equation}
\nn
\ees \in \argmax_{\eta\in\overline{\Theta_\eta}}\sum_{i=1}^{\left\lf T_n\right\rf} m\left(\hat{\ep}_i,\eta\right).
\end{equation}

\end{enumerate}

The rest of this paper is organized as follows.
We briefly summarize some prerequisites in Section \ref{yu:notation and assumption}, and then presents theoretical results in Section \ref{yu:main}.
Section \ref{sec:Implementation} introduces new classes and methods in \yuima \proglang{R} package (\cite{YUIMA14} and \cite{iacus2018simulation}) for the estimation procedure proposed in the previous sections.
Some numerical examples based on simulated and real data are given in Section \ref{sec:Numerical Examples}.

\section{Preliminaries}
\label{yu:notation and assumption}

\subsection{Notations and conventions}

Here are some basic notations and conventions used in this paper.

\begin{itemize}
\item For any vector variable $x=(x^{(i)})$, we write $\p_x=\left(\frac{\p}{\p x^{(i)}}\right)_i$.
\item $C$ denotes a universal positive constant which may vary at each appearance.
\item $\top$ stands for the transpose operator, and $v^{\otimes2}:= vv^\top$ for any matrix $v$.
\item For a matrix $M=(M_{ij})\in\mbbr^{d_1}\times\mbbr^{d_2}$ and vector $u=(u_j)\in\mbbr^{d_2}$, we will write $M[u]=\sum_{j=1}^{d_2} M_{\cdot j} u_j\in\mbbr^{d_1}$. In particular, when $d_1=1$, it stands for the dot product of two vectors. We will also write $M[U]=\sum_{k,l} M_{kl}U_{kl}$ for two square matrices $M$ and $U$ of the same order.
\item The convergences in probability and in distribution are denoted by $\cip$ and $\cil$, respectively, and all limits appearing below are taken for $n\to\infty$ unless otherwise mentioned.
\item For two nonnegative real sequences $(a_n)$ and $(b_n)$, we write $a_n \lesssim b_n$ if $\limsup_n(a_n/b_n)<\infty$.
\item For any $x\in\mbbr$, $\lf x \rf$ denotes the maximum integer which does not exceed $x$.
\item Given a function $h:\mathbb{R}\to\mathbb{R}^+$ and a signed measure $m$ on a one-dimensional Borel space, we write
\begin{equation}\nn
||m||_h=\sup\left\{|m(f)|:\mbox{$f$ is $\mathbb{R}$-valued, $m$-measurable and satisfies $|f|\leq h$}\right\}.
\end{equation}
\end{itemize}

\subsection{Basic assumptions}

Denote by $(\Omega,\mcf,(\mcf_t)_{t\ge 0},\pr)$ the underlying complete filtered probability space; every processes are adapted to the filtration $(\mcf_t)$. 
We will write $\E$ for the expectation operator associated with $\pr$.


\begin{Assumption}\label{Moments}
The {\lp} $J$ has moments of any order with $\E[J_1]=0$, $\E[J_1^2]=1$, and $\E[|J_1|^q]<\infty$, for any $q>0$.
\end{Assumption}

\begin{Assumption}\label{Smoothness}
\begin{enumerate}

\item 
The drift coefficient $a(\cdot,\al_0)$ and scale coefficient $c(\cdot,\gam_0)$ are Lipschitz continuous, and $c(x,\gam)\neq0$ for every $(x,\gam)$.

\item For each $i \in \left\{0,1\right\}$ and $k \in \left\{0,\dots,4\right\}$, the following conditions hold:
\begin{itemize}
\item The coefficients $a$ and $c$ admit 
the partial derivatives $\partial_x^i \partial_\alpha^k a$ and $\partial_x^i \partial_\gamma^k c$ for $i\ge 0$ and $k\ge 0$, all of which have continuous extensions as elements in $\mathcal{C}(\mathbb{R}\times\overline{\Theta})$.
\item There exists nonnegative constant $C_{i,k}$ satisfying
\begin{equation}\label{polynomial}
\sup_{(x,\alpha,\gamma) \in \mathbb{R} \times \overline{\Theta}_\alpha \times \overline{\Theta}_\gamma}\frac{1}{1+|x|^{C_{i,k}}}\left\{|\partial_x^i\partial_\alpha^ka(x,\alpha)|+|\partial_x^i\partial_\gamma^kc(x,\gamma)|+|c^{-1}(x,\gamma)|\right\}<\infty.
\end{equation}
\end{itemize}
\end{enumerate}
\end{Assumption}

Then 
\begin{Assumption}\label{Stability}
\begin{enumerate}
\item
There exists a probability measure $\pi_0$ such that for every $q>0$, we can find constants $a>0$ and $C_q>0$ for which 
\begin{equation}\label{Ergodicity}
\sup_{t\in\mathbb{R}_{+}} \exp(at) ||P_t(x,\cdot)-\pi_0(\cdot)||_{h_q} \leq C_qh_q(x),
\end{equation}
for any $x\in\mathbb{R}$ where $h_q(x):=1+|x|^q$.
\item 
For any $q>0$,  we have
\begin{equation}
\sup_{t\in\mathbb{R}_{+}}E[|X_t|^q]<\infty. 
\end{equation}
\end{enumerate}
\end{Assumption}

See \cite[Proposition 5.4]{Mas13-1} for easy-to-verify conditions for Assumption \ref{Stability}.

\medskip

We introduce a block-diagonal $p\times p$-matrix
\begin{equation}
\Gam=\diag(\Gam_\gam,\Gam_\al),
\nonumber
\end{equation}
whose components are defined by:
\begin{align*}
&\Gam_\gam:=2\int_\mbbr \frac{(\p_\gam c(x,\gam_0))^{\otimes2}}{c^2(x,\gam_0)}\pi_0(dx),\\
&\Gam_\al:=\int_\mbbr\frac{(\p_\al a(x,\al_0))^{\otimes2}}{c^2(x,\gam_0)}\pi_0(dx).
\end{align*}

\begin{Assumption}\label{nd}
$\Gam$ is positive definite.
\end{Assumption}

\medskip

We define real valued functions $\mbbh_1(\gam)$ and $\mbbh_{2}(\al)$ on $\Theta_\gam$ and $\Theta_\al$ by
\begin{align}
&\mbbh_1(\gam)=-\int_\mbbr \left(\log c^2(x,\gam)+\frac{c^2(x,\gam_0)}{c^2(x,\gam)}\right)\pi_0(dx),\label{KL1}\\
&\mbbh_2(\al)=-\int_\mbbr c(x,\gam_0)^{-2}(a(x,\al_0)-a(x,\al))^2\pi_0(dx). \label{KL2}
\end{align}

We assume the following identifiability condition for $\mbbh_1(\gam)$ and $\mbbh_2(\al)$:
\begin{Assumption}\label{Identifiability}
$\theta^\star\in\Theta$,
and there exist positive constants $\chi_\gam$ and $\chi_\al$ such that for all $(\gam,\al)\in\Theta$,
\begin{align}
&\label{idg}\mbby_1(\gam):=\mbbh_1(\gam)-\mbbh_1(\gam_0)\leq-\chi_\gam|\gam-\gam_0|^2,\\
&\label{ida}\mbby_2(\al):=\mbbh_2(\al)-\mbbh_2(\al_0)\leq-\chi_\al|\al-\al_0|^2.
\end{align}
\end{Assumption}

\subsection{Examples}
Although Assumption \ref{Moments} imposes a strong restriction on the mean and variance structure of $Z$, there is still room for statistical modeling of $Z$ with respect to, for example, its skewness and jump activity. 
In this section, with the parameter constraints for Assumption \ref{Moments}, we give concrete L\'{e}vy processes induced from subordinators (i.e. non-decreasing L\'{e}vy processes) by the following two procedures:
For given two independent subordinators $\tau^1$ and $\tau^2$, one can easily construct a possibly skewed L\'{e}vy process of finite variation by taking its bilateral version: $\tau':=\tau^1-\tau^2$.
Another way to deliver a new L\'{e}vy process is to take a normal mean variance mixture of a subordinator: for $\mu, \beta\in\mathbb{R}$, a subordinator $\tau$, and a standard normal random variable $\eta$ being independent of $\tau$, the normal mean variance mixture of $\tau$ at time $t$ is given by
\begin{equation}
\nn Z_t = \mu t+\tau_t\beta+\sqrt{\tau_t}\eta.
\end{equation}
By their construction, $E[|\tau'|^q]<\infty$ and $E[|Z|^q]<\infty$ hold as long as the $q$-th moment of $\tau^1, \tau^2$ and $\tau$ exists.
It is worth noting that having a generator of $\tau^1, \tau^2$ and $\tau$ in hand, we can directly obtain that of $\tau'$ and $Z$.
All of the following induced L\'{e}vy processes can be generated by the functions {\texttt{setLaw} and \texttt{simulate} in YUIMA package.

\begin{Example}(Bilateral gamma \cite{KucTap08-1})
For $\del_1,\gam_1,\del_2,\gam_2\in\mbbr_+$, the bilateral gamma process $\tau'$ is defined by the difference of two independent gamma subordinators $\tau^1$ and $\tau^2$ whose L\'{e}vy densities are expressed as: for each $i\in\{1,2\}$,
\begin{align*}
f_{\tau^i}(z)=\frac{\del_i}{z}e^{-\gam_iz}, \quad z>0.
\end{align*}
We write the law of $\tau'_1$ as $bgamma(\del_1,\gam_1,\del_2,\gam_2)$ and it is straightforward from the form of $f_{\tau^i}$ that $\tau'_t\sim bgamma(\del_1t,\gam_1,\del_2t,\gam_2)$.
Since the density function of $bgamma(\del_1,\gam_1,\del_2,\gam_2)$ is the convolution of two gamma density, it satisfies the symmetry relation
\begin{equation*}
p(x; \del_1,\gam_1,\del_2,\gam_2)=p(-x;\del_2,\gam_2, \del_1,\gam_1), \quad x\in\mbbr\setminus\{0\},
\end{equation*}
and on the positive real line, its form is given by
\begin{equation*}
p(x; \del_1,\gam_1,\del_2,\gam_2)=\frac{\gam_1^{\del_1}\gam_2^{\del_2}}{(\gam_1+\gam_2)^{\frac{1}{2}(\del_1+\del_2)}\Gam(\del_1)}x^{\frac{1}{2}(\del_1+\del_2)-1}e^{-\frac{x}{2}(\gam_1-\gam_2)}W_{\frac{1}{2}(\del_1-\del_2), \frac{1}{2}(\del_1+\del_2-1)}(x(\gam_1+\gam_2)),
\end{equation*}
where $W_{\lam, \mu}(z)$ denotes the Whittaker function.
By using the independence between $\tau^1$ and $\tau^2$, the parameter constraints for Assumption \ref{Moments} are written as follows:
\begin{equation}
\nn E[Z_1]= \frac{\del_1}{\gam_1}-\frac{\del_2}{\gam_2}=0, \quad V[Z_1]= \frac{\del_1}{\gam^2_1}+\frac{\del_2}{\gam^2_2}=1.
\end{equation}

\end{Example}
\begin{Example}(Normal (exponentially) tempered stable) The normal (exponentially) tempered stable law $NTS(\al, a, b, \beta, \mu)$ is defined by the law of the normal mean variance mixture of the positive exponentially tempered stable random variable whose L\'{e}vy density is given by
\begin{equation}
\nn f(z)=az^{-1-\al}e^{-bz}, \quad \al\in(0,1), \ a>0, \ b>0,
\end{equation}
and its law is denoted by $TS(\al,a,b)$; especially inverse Gaussian law corresponds to $\al=\frac{1}{2}$.
From \cite[Theorem 30.1]{Sat99}, the L\'{e}vy density of $NTS(\al, a, b, \beta, \mu)$ is explicitly expressed as:
\begin{equation}
\nn g(z)=\sqrt{\frac{2}{\pi}}ae^{\beta z}\left(\frac{z^2}{2b+\beta^2}\right)^{-\frac{\al}{2}-\frac{1}{4}}K_{\al+\frac{1}{2}}\left(z\sqrt{(2b+\beta^2)}\right), 
\end{equation}
where $K_{\al+\frac{1}{2}}$ stands for the modified Bessel function of the third kind with index $\al+\frac{1}{2}$.
From the expression of the L\'{e}vy density, the associated positive exponentially tempered stable subordinator $\tau$ satisfies $\tau_t\sim TS(\al,at,b)$ and thus the corresponding normal (exponentially) tempered stable process $Z$ also does $Z_t\sim NTS(\al, at, b, \beta, \mu t)$. 
Since $z^\lambda K_\lambda(z)\lesssim 1$ as $z\downarrow0$ for $\lambda>0$, the Blumenthal-Getoor index of $Z$ is $2\al$.
The parameter constraints for Assumption \ref{Moments} are written as follows:
\begin{equation}
\nn E[Z_1]= \mu-a\al \Gamma(-\alpha)b^{\al-1}\beta=0, \quad V[Z_1]=a\al \Gam(-\al)b^{\al-1}\left[\frac{(\al-1)\beta^2}{b}-1\right]=1.
\end{equation}
Especially in the simple case where $\mu=\beta=0$ (that is, $Z$ is a time-changed Brownian motion), the above constraints are reduced to 
\begin{equation}
\nn-a\al b^{\al-1}\Gam(-\al)=1.
\end{equation}
\end{Example}

\color{black}
\subsection{Stepwise Gaussian quasi-likelihood estimation}

Here and in what follows, for any process $Y$ we will denote by $\D_j Y$ the $j$-th increment
\begin{equation}
\D_j Y := Y_{t_{j}}-Y_{t_{j-1}},
\nonumber
\end{equation}
and $f_{j-1}(\theta):=f(X_{t_{j-1}},\theta)$ for a measurable function $f$ on $\mbbr\times\overline{\Theta}$.
Building on the discrete time formal Gaussian approximation, we define the stepwise GQL functions $\mbbh_{1,n}(\gam)$ and $\mbbh_{2,n}(\al)$ as follows \cite{MasUeh17-2}:
\begin{align}
&\nn\mbbh_{1,n}(\gam):=-\frac{1}{2T_n}\sumj \left(h \log c^2_{j-1}(\gam)+\frac{(\D_j X)^2}{c^2_{j-1}(\gam)}\right),\\
&\nn\mbbh_{2,n}(\al):=-\frac{1}{2T_n}\sumj \frac{(\D_j X-h a_{j-1}(\al))^2}{h c^2_{j-1}(\ges)},
\end{align}
where $\ges$ is any maximizer of $\mbbh_{1,n}$ over $\overline{\Theta}_{\gam}$.
We then define the associated stepwise GQMLE $\tes:=(\ges,\aes)$ where $\aes$ is any maximizer of $\mbbh_{2,n}$ over $\overline{\Theta}_{\al}$.
Formally, the first-stage $\mbbh_{1,n}(\gam)$ corresponds to the quasi log-likelihoods associated with the (fake) approximation $N(x,\,c^2(x,\gam)h)$ for $\mcl(X_{t_j}|X_{t_{j-1}}=x)$, and also the second-stage $\mbbh_{2,n}(\al)$ does to the one associated with $N(x+a(x,\al)h,\,c^2(x,\ges)h)$.

Let $\nu_0(dz)$ denote the (true) {\lm} of $J$, and
\begin{equation}
\nu_k := \int z^k\,\nu_0(dz),\qquad k\ge 2.
\nonumber
\end{equation}
Under the aforementioned Assumptions \ref{Moments} to \ref{Identifiability}, both $\nu_3$ and $\nu_4$ exist and are finite (and $E[J_1^2]=\nu_2 + \sig^2=1$, where $\sig^2 \ge 0$ denotes the Gaussian variance of $J$, possibly $\sig^2=0$).
We can deduce the asymptotic normality and the uniform tail-probability estimate:

\begin{Thm}
\label{hm:thm_old}
Suppose that Assumptions \ref{Moments} to \ref{Identifiability} holds.
\begin{enumerate}
\item The GQMLE $\tes=(\ges,\aes)$ satisfies
\begin{equation}
\sqrt{T_n}(\tes-\tz)\cil N_{p_\gam + p_\al}\left(0,\Gam^{-1}\Sig_\theta(\Gam^{-1})^\top\right),
\end{equation}
where the $p\times p$-matrix
\begin{equation}
\Sig_\theta = \begin{pmatrix}\Sig_\gam &\Sig_{\gam,\al}\\ \Sig_{\gam,\al}^\top & \Sig_{\al}\end{pmatrix},
\nonumber
\end{equation}
is defined by
\begin{align}
    \nn &\Sig_{\gam} = \nu_4 \, \int \left(\frac{\p_\gam c(x,\gam_0)}{c(x,\gam_0)}\right)^{\otimes2} \pi_0(dx),\\
    \nn &\Sig_{\gam,\al} = \nu_3\,\int \frac{\p_\gam c(x,\gam_0)\p_\al a(x,\al_0)}{c^2(x,\gam_0)} \pi_0(dx), \\
    \nn &\Sig_\al = \int \left(\frac{\p_\al a(x,\al_0)}{c(x,\gam_0)}\right)^{\otimes2} \pi_0(dx).
\end{align}

\item For any $L>0$, there exists a constant $C_L>0$ such that
\begin{equation}\label{eq: TPE}
\nn \sup_{n\in\mbbn} P\left(\left|\sqrt{T_n}(\tes-\tz)\right|>r\right)\leq \frac{C_L}{r^L}, \qquad r>0.
\end{equation}
\end{enumerate}
\end{Thm}

We note that $\Sig_{\gam,\al}$, $\aes$ and $\ges$ are asymptotically independent if $\nu_3=0$, hence if in particular $\nu$ is symmetric.
We refer to \cite{Mas13-1} and \cite{MasUeh17-2} for technical details of the proof of Theorem \ref{hm:thm_old}; 
although the two cited papers used a $Z$-estimator type identifiability condition, which is seemingly different from Assumption \ref{Identifiability} ($M$-estimator type), it is trivial that we can follow the same line without any essential change. 

\medskip


To appreciate the difficulty of relaxing the standing condition $nh^2\to 0$ in \eqref{sampling.design}, let us first mention the case of diffusions:
let $(w_t)$ be a standard Wiener process, and consider the following one-dimensional diffusion process
\begin{equation*}
dY_t=a(Y_t)dt+b(Y_t)dw_t
\end{equation*}
defined on the stochastic bases $(\Omega, \mcf, (\mcf_t)_{t\in\mbbr_+}, P)$.
Write the infinitesimal generator of $Y$ as $\mca$.
By repeatedly applying It\^{o}'s formula, it follows that for $0\leq s<t$, $q\in\mbbn$, and a sufficiently smooth function $f$, we have
\begin{align*}
&f(Y_t)\\
&=f(Y_s)+\int_s^t \p f(Y_u)dY_u+\frac{1}{2}\int_s^t \p^2 f(Y_u)b^2(Y_u)du\\
&=f(Y_s)+\int_s^t \left(\p f(Y_u)a(Y_u)+\frac{1}{2}\p^2 f(Y_u)b^2(Y_u)\right)du+\int_s^t \p f(Y_u)b(Y_u)dw_u\\
&=f(Y_s)+\int_s^t \mca f(Y_u)du+\int_s^t \p f(Y_u)b(Y_u)dw_u\\
&=f(Y_s)+\int_s^t \left(\mca f(Y_s)+\int_u^s \mca^2 f(Y_v)dv+\int_u^s\p\mca f(Y_v)b(Y_v)dw_v\right)du+\int_s^t \p f(Y_u)b(Y_u)dw_u\\
&=f(Y_s)+(t-s)\mca f(Y_s)+\int_s^t \left(\int_u^s \mca^2 f(Y_v)dv+\int_u^s\p\mca f(Y_v)b(Y_v)dw_v\right)du+\int_s^t \p f(Y_u)b(Y_u)dw_u\\
&=\sum_{i=0}^{q-1}\frac{(t-s)^i}{i!}\mca^i f(Y_s)+\int_s^t\int\dots\int \mca^qf(Y_{v_q})dv_1\dots dv_q+\text{(martingale term)}.
\end{align*}
Thus, under suitable integrability conditions, we obtain the expansion of $E[f(Y_t)|\mcf_s]$:
\begin{equation}
E[f(Y_t)|\mcf_s]=\sum_{i=0}^{q-1}\frac{(t-s)^i}{i!}\mca^i f(Y_s)+O_p((t-s)^q).
\label{hm:ho-ItoTaylor}
\end{equation}
In particular, the first-order approximation of the conditional expectation and conditional variance are given by
\begin{align}
E[X_t|\mcf_s] &\approx X_s+(t-s)a(X_s,\al_0), \nn\\
V[X_t|\mcf_s] &\approx E[(X_t-X_s-(t-s)a(X_s,\al_0))^2|\mcf_s] \approx (t-s)c(X_s,\gam_0),
\nonumber
\end{align}
respectively, which are used for constructing the GQL for diffusions.
Relaxation of the condition $nh^2\to0$ to $nh^k\to0$ for $k>2$ is then possible by taking $q=q(k)$ large enough according to the value of $k$:
the associated GQMLE has the consistency and asymptotic normality under $nh^k\to0$.
We refer to \cite{Kes97} for details.
At this point, we should remark that 
the GQL based on the first term on the right-hand side of \eqref{hm:ho-ItoTaylor} (with $f(y)=y$ and $f(y)=y^2$) is fully explicit whatever $k>2$ is.

On the other hand, although similar It\^{o}-Taylor expansions to $E[X_t|\mcf_s]$ and $V[X_t|\mcf_s]$ can be easily derived in our L\'{e}vy driven case \eqref{hm:sde}, the corresponding infinitesimal generator contains not only the differential operator but also the integral operator with respect to the L\'{e}vy measure $\nu_\eta$ of the driving L\'{e}vy noise.
Specifically, its infinitesimal generator $\tilde{\mca}$ is given by
$$\tilde{\mca} f(x)=a(x,\al)\p f(x)+\int (f(x+c(x,\gam)z)-f(x)-\p f(x)c(x,\gam)z)\nu_\eta(dz),$$
for a suitable function $f$.
Consequently, the modified GQL based on the higher-order It\^{o}-Taylor expansion \eqref{hm:ho-ItoTaylor} contains the unknown parameter $\eta$ in addition to the drift and scale parameters.
It is not clear that the simultaneous estimation of $\al,\gam,$ and $\eta$ by the modified GQL has a nice theoretical property.
Even if it does, the entailed numerical optimization involved would be quite heavy and unstable since, for each $\eta$, we need to repeatedly compute several integrals with respect to $\nu_\eta$ inside of the modified GQL.
For this reason, it is difficult to remove the condition $nh^2\to0$ in the present general non-linear-SDE setting, as long as we use the GQL based on the stochastic expansion \eqref{hm:ho-ItoTaylor}.


\medskip

In practice, the sampling points $t_1,\dots t_n$ may not be equally spaced.
In such a case, supposing \eqref{sampling.design}, we remark that the same statement as in Theorem \ref{hm:thm_old} remains in place under the additional ``sampling-balance'' condition:
\begin{equation}
\frac{\min_{1\leq j\leq n}(t_j-t_{j-1})}{\max_{1\leq j\leq n}(t_j-t_{j-1})} \to 1.
\nonumber
\end{equation}
For more technical details, see the discussion in \cite[p. 1604--1605]{Mas13-1}.

\section{Theoretical results}
\label{yu:main}

\subsection{Stochastic expansion of residual functional}

Having the GQMLE in hand, we now turn to approximating $\mcl(J_1)$, the the unit-time distribution of $J$, based on the residuals
\begin{equation}
\nn \hat{\ep}_i = \sum_{j\in A_i} \widehat{\D_j J},
\end{equation}
where
\begin{equation}
\nn \widehat{\D_j J}=\frac{\D_j X-h a_{j-1}(\aes)}{c_{j-1}(\ges)}.
\end{equation}
Write $\hat{u}_{\al,n} = \nn \sqrt{T_n}(\aes-\al_0)$ and $\hat{u}_{\gam,n} = \nn \sqrt{T_n}(\ges-\gam_0)$, and let
\begin{equation}
\hat{u}_{\theta,n}:=(\hat{u}_{\gam,n},\hat{u}_{\al,n}).
\nonumber
\end{equation}
From now on we will mostly omit ``$(\tz)$'' from notation. In particular, for a measurable function $f(x,\theta)$ we will abbreviate $f_{j-1}(\tz)$ as $f_{j-1}$.

\begin{Thm}\label{yu:se}
Suppose that Assumptions \ref{Moments} to \ref{Identifiability} hold. 
Let $\rho:\,\mbbr\to\mbbr$ be a $\mcc^{2}$-function such that
\begin{equation}
\max_{i\in\{0,1,2\}}\left|\p^i_\ep \rho(\ep) \right|\lesssim 1+|\ep|^{C}.
\nonumber
\end{equation}
Then, we have
\begin{equation}
\sum_{i=1}^{\left\lf T_n\right\rf} \rho\left(\hat{\ep}_i\right) 
=\sum_{i=1}^{\left\lf T_n\right\rf} \rho\left(\ep_i\right) 
+ \frac{1}{\sqrt{T_n}} \sum_{i=1}^{\left\lf T_n\right\rf} \p_\ep \rho\left(\ep_i\right) b_i [\hat{u}_{\theta,n}]
+o_p\left(\sqrt{T_n}\right),
\label{yu:se-1}
\end{equation}
where the random sequence $(b_i)_{i=1}^{\left\lf T_n\right\rf}\subset\mbbr^p$ is given by
\begin{equation}
b_i=b_i(\tz)  := \sum_{j\in A_i} \begin{pmatrix}\p_\gam (c^{-1})_{j-1} (\D_j X-h a_{j-1}) \\
-h c^{-1}_{j-1} \p_\al a_{j-1}\end{pmatrix}.
\nonumber
\end{equation}
\end{Thm}


Theorem \ref{yu:se} reveals the quantitative effect of plugging in $\tes$.
The second term on the right-hand side of \eqref{yu:se-1} is not $o_p(\sqrt{T_n})$ but $O_p(\sqrt{T_n})$, which implies that the asymptotic distribution of $\ees$ is indeed subject to influence of the proposed unit-time approximation; this is natural and expected, for we are using the $\sqrt{T_n}$-consistent (generally sub-optimal) estimator of $\theta$.

\medskip

To prove Theorem \ref{yu:se}, we begin with a preliminary estimate.

\begin{Lem}\label{yu:gap}
Under Assumptions \ref{Moments} to \ref{Identifiability}, for each $r\ge 2$ we have
\begin{align}
& \max_{1\le i\le \lf T_n\rf} \E\left[\left| \hat{\ep}_i-\ep_i - \frac{1}{\sqrt{T_n}} b_i [\hat{u}_{\theta,n}] \right|^r\right] \lesssim h \vee T_n^{-r}.
\label{yu:gap-1}
\end{align}
\end{Lem}

\begin{proof}
We will abbreviate $a(X_s,\al_0)$ as $a_s$ and so on; with a slight abuse of notation, we will write $a_{j-1}$ for $a_{t_{j-1}}$.
Write $\widehat{\D_j J}=\del_j(\tes)$; then,
\begin{equation}
\del_j(\tes) = \del_j + (\p_\theta\del_j) [\tes-\tz] + \frac12 \left(\p_\theta^2\del_j(\tz+s_n(\tes-\tz))\right)[(\tes-\tz)^{\otimes 2}]
\nonumber
\end{equation}
for a suitable random point $s_n\in[0,1]$.
Decompose $\hat{\ep}_i-\ep_i$ as follows:
\begin{align}
\hat{\ep}_i-\ep_i &= \sum_{j\in A_i}
(\widehat{\D_j J}-\D_j J)+\left(J_{(\lf\frac{i-1}{h}\rf  + 1)h} - J_{i-1}\right)+\left(J_i - J_{\lf\frac{i}{h }\rf h }\right) \nn\\
&=\frac{1}{\sqrt{T_n}} b_i [\hat{u}_{\theta,n}]+\zeta_{1,i}+\zeta_{2,i}+\zeta_{3,i},
\nonumber
\end{align}
where
\begin{align*}
&\zeta_{1,i}:= \frac{1}{2T_n} \sum_{j\in A_i} \left(\p_\theta^2\del_j(\tz+s_n(\tes-\tz))\right)[\hat{u}_{\theta,n}^{\otimes 2}],
\\
&\zeta_{2,i}:=\sum_{j\in A_i}\left(\del_j - \D_j J\right) 
= \sum_{j\in A_i} c_{j-1}^{-1}\left( \intj (a_s-a_{j-1})ds+\intj (c_{s-}-c_{j-1})dJ_s\right),
\\
&\zeta_{3,i}:=\left(J_{(\lf\frac{i-1}{h}\rf  + 1)h} - J_{i-1}\right)+\left(J_i - J_{\lf\frac{i}{h }\rf h }\right).
\end{align*}
Before proceeding, let us note that by Theorem \ref{hm:thm_old} the sequence $(\hat{u}_{\theta,n})$ is $L^r$-bounded for each $r\ge 2$:
\begin{equation}
\E\left(|\hat{u}_{\theta,n}|^r\right) \lesssim 1,
\label{yu:gap-p1}
\end{equation}
where we implicitly assume that $r\ge 1$ when using this notation.

First we will deduce
\begin{equation}
\max_{1\le i\le \lf T_n\rf} \E\left(\left| \zeta_{1,i} \right|^r\right) \lesssim T_n^{-r}.
\label{yu:gap-p2}
\end{equation}

Since the parameter space $\Theta=\Theta_\al\times\Theta_\gam$ is supposed to be bounded and convex, the Sobolev inequality is in force (see \cite{Ada73} for details): for a random field $u\in\mcc^1(\Theta)$ and $q>p$, we have
$$E\left[\sup_{\theta\in\Theta}|u(\theta)|^q\right] \lesssim \sup_{\theta\in\Theta}\left\{E[|u(\theta)|^q]+E[|\p_\theta u(\theta)|^q]\right\}.$$

Noting the identities
\begin{align}
\p_\gam^k \p_\al^l \del_j(\theta) &= - h \left\{\p_\gam^k (c^{-1})_{j-1}(\gam)\right\} \p_\al^l a_{j-1}(\al), \nn\\
\p_\gam^k \del_j(\theta) &= \left\{\p_\gam^k (c^{-1})_{j-1}(\gam)\right\} (\D_j X - h a_{j-1}(\al)),
\nonumber
\end{align}
valid for each $k\ge 0$ and $l\ge 1$,
we can apply Sobolev's and Jensen's inequalities to conclude that, for $r>p_\gam$,
\begin{align*}
&\E\left[\left| \sum_{j\in A_i} \left(\p_\theta^2\del_j(\tz+s_n(\tes-\tz))\right) \right|^r \right]\\
&\lesssim \E\left[\sup_{\gam\in\Theta_\gam}\left|
\sum_{j\in A_i} \p_\gam^2 (c^{-1})_{j-1}(\gam) \intj c_{s-}dJ_s\right|^r\right]
+ \E\left[ h \sum_{j\in A_i} (1+|X_{t_{j-1}}|^C)\right]
\\
&\lesssim \max_{k\in\{2,3\}}\sup_{\gam\in\Theta_\gam}
\E\left[\left| \sum_{j\in A_i} \p_\gam^k (c^{-1})_{j-1}(\gam) \intj c_{s-}dJ_s\right|^r\right] +1.
\end{align*}
Let $\chi_j(s)$ denote the indicator function of the interval $[t_{j-1},t_{j})$.

To proceed, we recall Burkholder's inequality for stochastic integrals with respect to a centered {\lp}: 
under the moment conditions on $J_\eta$, for any predictable process $H$ and $q\geq 2$ we have
$$E\left[\left|\int_{\left\lf \frac{i-1}{h }\right\rf h }^{\left\lf \frac{i}{h }\right\rf h } H_sdJ_s\right|^q\right]
\leq K_q(\nu_{\eta,2}^{q/2}+\nu_{\eta,q})\int_{\left\lf \frac{i-1}{h }\right\rf h }^{\left\lf \frac{i}{h }\right\rf h } E\left[|H_s|^q\right]ds
\lesssim \int_{\left\lf \frac{i-1}{h }\right\rf h }^{\left\lf \frac{i}{h }\right\rf h } E\left[|H_s|^q\right]ds,$$
where $K_q$ is a positive constant depending only on $q$, and $\nu_{\eta,k} := \int z^k\,\nu_\eta(dz)$ for $k\geq 2$ (see \cite[Theorem IV 48]{Pro04}).

Then, we see that the last expectation equals
\begin{align}
\label{yu:burk}\E\left[\left|\int_{\left\lf \frac{i-1}{h }\right\rf h }^{\left\lf \frac{i}{h }\right\rf h }\sum_{j\in A_i}\chi_j(s)
\p_\gam^k (c^{-1})_{j-1}(\gam) c_{s-}dJ_s\right|^r\right]
\lesssim \sum_{j\in A_i}\intj \E\left[\left| \p_\gam^k (c^{-1})_{j-1}(\gam)\, c_{s}\right|^r\right]ds 
\lesssim 1.
\end{align}
This together with \eqref{yu:gap-p1} concludes \eqref{yu:gap-p2}.

Turning to $\zeta_{2,i}$, we note the standard moment estimate: for any real $r\ge 2$,
\begin{equation}
\max_{j\le n}\sup_{s\in(t_{j-1},t_j]}\E^{}\left[|X_{s}-X_{t_{j-1}}|^r \right] \lesssim h.
\nonumber
\end{equation}
With this and the Lipschitz property of $x\mapsto (a(x,\al_0), c(x,\gam_0))$,  analogous arguments as in handling $\zeta_{1,i}$ yield that for each $r\geq2$
\begin{equation}
\max_{1\le i\le \lf T_n\rf} \E\left[\left| \zeta_{2,i} \right|^r\right] \lesssim h.
\label{yu:gap-p3}
\end{equation}

As for the remaining $\zeta_{3,i}$, it follows from Assumption \ref{Moments} and the stationarity of increments that for each $r\ge 2$,
\begin{align}
\max_{1\le i\le \lf T_n\rf} \E\left[\left|\zeta_{3,i}\right|^r\right]
\lesssim \E\left[\left(J_{j-1}-J_{\lf\frac{j-1}{h }\rf h +h }\right)^r\right]+\E\left[\left(J_{\lf\frac{j}{h }\rf h }-J_j\right)^r\right]\lesssim  h.
\label{yu:gap-p4}
\end{align}
Piecing together \eqref{yu:gap-p2}, \eqref{yu:gap-p3}, and \eqref{yu:gap-p4} concludes the proof.
\end{proof}

\begin{proof}[Proof of Theorem \ref{yu:se}]
Mimicking the estimates for \eqref{yu:gap-p2}, for each $r\geq 2$ we obtain
\begin{equation*}
\max_{1\le i\le \lf T_n\rf} \E\left[\left| \frac{1}{\sqrt{T_n}} b_i \right|^r\right] \lesssim T_n^{-r/2}.
\end{equation*}
Combined with \eqref{yu:gap-1} and \eqref{yu:gap-p1}, it follows from H\"{o}lder's inequality that for each $r\ge 2$, 
\begin{equation}
\max_{1\le i\le \lf T_n\rf} \E\left[\left| \hat{\ep}_i-\ep_i\right|^r\right] \lesssim h \vee T_n^{-\frac{r}{2}},
\label{yu:gap-2}
\end{equation}
and hence $\max_{1\le i\le \lf T_n\rf} \E[|\hat{\ep}_i|^r]\lesssim 1$ as well.
We use the expression
\begin{align*}
\sum_{i=1}^{\left\lf T_n\right\rf} \rho\left(\hat{\ep}_i\right)
=\sum_{i=1}^{\left\lf T_n\right\rf} \rho\left(\ep_i\right)+\sum_{i=1}^{\left\lf T_n\right\rf} \p_\ep \rho\left(\ep_i\right)(\hat{\ep}_i-\ep_i)
+\frac{1}{2}\sum_{i=1}^{\left\lf T_n\right\rf} \p_\ep^2 \rho\left(\ep_i+u\left(\hat{\ep}_i-\ep_i\right)\right) (\hat{\ep}_i-\ep_i)^2
\end{align*}
for a (random) $u\in[0,1]$.
By means of Schwarz's inequality and \eqref{yu:gap-1},
\begin{align}
& \E\left[\left|\sum_{i=1}^{\left\lf T_n\right\rf} \p_\ep \rho\left(\ep_i\right)(\hat{\ep}_i-\ep_i)
- \frac{1}{\sqrt{T_n}} \sum_{i=1}^{\left\lf T_n\right\rf} \p_\ep \rho\left(\ep_i\right) b_i [\hat{u}_{\theta,n}]\right|\right]
\nn\\
&=\E\left[\left|
\sum_{i=1}^{\left\lf T_n\right\rf} \p_\ep \rho\left(\ep_i\right) \left( \hat{\ep}_i-\ep_i - \frac{1}{\sqrt{T_n}} b_i [\hat{u}_{\theta,n}] \right)
\right|\right]
\lesssim T_n \sqrt{h \vee T_n^{-2}} = \sqrt{T_n} \sqrt{nh^2 \vee T_n^{-1}} = o\left(\sqrt{T_n}\right).
\nonumber
\end{align}
By the moment estimates in the proof of Lemma \ref{yu:gap}, H\"{o}lder's inequality, and \eqref{yu:gap-2}, and also recalling \eqref{sampling.design}, we see that for $\del\in(1,2]$,
\begin{align}
& \E\left[\left|
\sum_{i=1}^{\left\lf T_n\right\rf} 
\p_\ep^2 \rho\left(\ep_i+u\left(\hat{\ep}_i-\ep_i\right)\right)(\hat{\ep}_i - \ep_i)^2 \right|\right] \nn\\
&\lesssim 
T_n \max_{1\le i\le \lf T_n\rf} \E\left[\left| \hat{\ep}_i-\ep_i\right|^{2\del}\right]^{1/\del}
\lesssim \sqrt{T_n} \left\{(nh^{1+2/\del})^{\del/2} \vee T_n^{-\del/2}\right\}^{1/\del} = o\left(\sqrt{T_n}\right).
\nonumber
\end{align}
This completes the proof.
\end{proof}

\subsection{$M$-estimation of noise parameter}

We keep Assumptions \ref{Moments} to \ref{Identifiability} in force.
Having Theorem \ref{yu:se} in hand, we proceed with estimation of $\eta$ based on the unit-time residual sequence $(\hat{\ep}_i)$. Let
\begin{equation}
\mbbh_{3,n}(\eta) :=\frac{1}{T_n}\sum_{i=1}^{\lf T_n\rf}m(\hat{\ep}_i,\eta),
\nonumber
\end{equation}
and consider an $M$-estimator
\begin{equation}
\label{hm:est.func.eta}
\ees \in \argmax_{\eta\in\overline{\Theta_\eta}} \mbbh_{3,n}(\eta).
\end{equation}
Among others, this includes the (quasi) maximum-likelihood for $m(\ep,\eta)=:\log f(\ep;\eta)$, where $\{f(\ep;\eta):\,\eta\in\Theta_\eta\}$ is a model for the unit-time noise distribution $\mcl(J_1)$.
We need to impose several conditions on the function $m$, all of which are standard in the general theory of $M$-estimation.

\begin{Assumption}\label{estimating}
$\ $\begin{enumerate}
\item $m\in\mcc^{2,3}(\mbbr \times \Theta_\eta)$ and $\ds{\max_{i\in\{0,1,2\} \atop k\in\{0,1,2,3\}}\sup_{\eta\in\overline{\Theta_\eta}}\left|\p^i_\ep \p^k_\eta m(\ep,\eta)\right|\lesssim 1+|\ep|^{C}}$ for some $C\ge 0$.
\item $\E\left[\p_\eta m(J_1,\eta_0)\right]=0$, the $p_\eta\times p_\eta$-matrix 
$\E\left[\left(\p_\eta m(J_1,\eta_0)\right)^{\otimes2}\right]$ 
is positive definite, and
\begin{equation}
\{\eta_0\} = \argmax_\eta \E\left[ m(J_1,\eta) \right].
\nonumber
\end{equation}
\end{enumerate}
\end{Assumption}

The consistency of $\ees$ can be easily seen from Theorem \ref{yu:se} and Lemma \ref{yu:gap}:
we have the continuous random function
\begin{equation}
\mbby_n(\eta) := \mbbh_{3,n}(\eta) - \mbbh_{3,n}(\eta_0) 
\cip \mbby(\eta) := \E\left[ m(J_1,\eta) \right] - \E\left[ m(J_1,\eta_0) \right],
\nonumber
\end{equation}
where $\mbby(\eta)\le 0$ by Jensen's inequality with $\mbby(\eta)=0$ if and only if $\eta=\eta_0$, and moreover, the convergence is uniform in $\eta\in\overline{\Theta_\eta}$ since $\sup_n\E[\sup_\eta |\p_\eta\mbby_n(\eta)|]<\infty$.
Hence the consistency $\ees\cip\eta_0$ follows.

\medskip

We turn to the asymptotic normality of $\ees$.
Let $\hat{u}_{\eta,n}:=\sqrt{T_n}(\ees-\eta_0)$.
In the sequel, for any measurable function $f(x,\theta)$ we will simply write $\hat{f}_{j-1}$ for $f_{j-1}(\tes)$.
\begin{Cor}\label{yu:clteta}
Under Assumptions \ref{Moments} to \ref{Identifiability}, and Assumption \ref{estimating}, we have
\begin{align}
\left( -\p_\eta^2\mbbh_{3,n}(\ees)\right) [\hat{u}_{\eta,n}] -\frac{1}{T_n} \sum_{i=1}^{\left\lf T_n\right\rf} \left(\p_\eta\p_\ep m\left(\hat{\ep}_i,\ees\right) \left(\hat{b}_i[\hat{u}_{\theta,n}] \right) \right)
&\cil N_{p_\eta}\left(0, \E\left[\left(\p_\eta m(J_1,\eta_0)\right)^{\otimes2}\right]\right),
\nonumber
\end{align}
where 
\begin{equation}
\hat{b}_i := \sum_{j\in A_i} \begin{pmatrix}\widehat{\p_\gam (c^{-1})}_{j-1} (\D_j X-h \hat{a}_{j-1}) \\
-h \widehat{c^{-1}}_{j-1}\widehat{\p_\al a}_{j-1}\end{pmatrix}.
\nonumber
\end{equation}
\end{Cor}
\color{black}

\begin{proof}
By the consistency $\ees\cip\eta_0$ we may and do focus on the event $\{\p_\eta\mbbh_{3,n}(\ees)=0\}$, on which
\begin{equation}
\left( - \int_0^1 \p_\eta^2\mbbh_{3,n}(\eta_0+s(\ees-\eta_0))ds\right) [\hat{u}_{\eta,n}] = \sqrt{T_n} \p_\eta\mbbh_{3,n}.
\nonumber
\end{equation}
Then, by Theorem \ref{yu:se} we have
\begin{align}
\left( -\p_\eta^2\mbbh_{3,n}(\ees) + o_p(1) \right) [\hat{u}_{\eta,n}] 
&= \frac{1}{\sqrt{T_n}}\sum_{i=1}^{\left\lf T_n\right\rf} \p_\eta m\left(\ep_i,\eta_0\right)
\nn\\
&{}\qquad
+ \frac{1}{T_n} \sum_{i=1}^{\left\lf T_n\right\rf} \left(\p_\eta\p_\ep m\left(\hat{\ep}_i,\ees\right) \left(\hat{b}_i[\hat{u}_{\theta,n}] \right) \right)
+o_p(1).
\nonumber
\end{align}
Hence, letting
\begin{equation}
\hat{u}_{n} := \left(\hat{u}_{\theta,n},\hat{u}_{\eta,n}\right),
\nonumber
\end{equation}
we have
\begin{equation}\label{yu:score}
\hat{H}_n [\hat{u}_n]
= \frac{1}{\sqrt{T_n}}\sum_{i=1}^{\left\lf T_n\right\rf} \p_\eta m\left(\ep_i,\eta_0\right) + o_p(1),
\end{equation}
where $\hat{H}_n$ is $p_\eta \times (p+p_\eta)$-matrix given by
\begin{equation}
\hat{H}_n = \left(
- \frac{1}{T_n} \sum_{i=1}^{\left\lf T_n\right\rf} \p_\eta\p_\ep m\left(\hat{\ep}_i,\ees\right) \hat{b}_i ,~
 -\p_\eta^2\mbbh_{3,n}(\ees)
 \right),
\nonumber
\end{equation}
and is $O_p(1)$.
Since the random variables $\ep_1,\ep_2,\dots,\ep_{\lf T_n\rf}$ are i.i.d., the desired result follows on applying the central limit theorem to \eqref{yu:score}.
\end{proof}

Corollary \ref{yu:clteta} suggests that the effect of plugging in the GQMLE does remain in the limit.
Therefore, in order to construct confidence interval and hypothesis testing for $\eta$, we need to verify the joint asymptotic distribution of $\hat{\mci}_n[\hat{u}_n]$ with some invertible matrix $\hat{\mci}_n$.
We additionally introduce the following condition.

\begin{Assumption}\label{yu:cov}
There exist $p_\eta\times p_\gam$-matrix $\Sigma_{\eta,\gam}$ and $p_\eta\times p_\al$-matrix $\Sigma_{\eta,\al}$ such that 
\begin{align}
&\label{yu:cr1} \frac{1}{T_n}\sum_{i=1}^{\lf T_n\rf}\p_\eta m(\ep_i, \eta_0)\sum_{j\in A_i} \left(\frac{\p_\gam c_{j-1}}{c_{j-1}}\left[(\D_j J)^2-h \right]\right)\cip \Sigma_{\eta,\gam},\\
&\label{yu:cr2} \frac{1}{T_n}\sum_{i=1}^{\lf T_n\rf}\p_\eta m(\ep_i, \eta_0)\sum_{j\in A_i}\left(\frac{\p_\al a_{j-1}}{c_{j-1}}\D_j J\right)  \cip \Sigma_{\eta,\al}.
\end{align}
Furthermore, the $(p+p_\eta)\times(p+p_\eta)$-matrix
\begin{equation*}
\Sig = \begin{pmatrix}\Sig_\gam & \Sig_{\gam,\al} & \Sig_{\eta,\gam}\\\Sig_{\gam,\al}^\top & \Sig_{\al} &\Sig_{\eta,\al}\\ \Sig_{\eta,\gam}^\top& \Sig_{\eta,\al}^\top & E\left[\left(\p_\eta m(J_1,\eta_0)\right)^{\otimes2}\right] \end{pmatrix},
\end{equation*}
is invertible.
\end{Assumption}

From now on, we will write $E^{i-1}[\cdot]$ for the conditional expectation $E[\cdot | \mcf_{i-1}]$.

\begin{Rem}\label{yu:cond}
By an elementary application of Burkholder's 
inequality, it is easy to deduce that \eqref{yu:cr1} and \eqref{yu:cr2} are equivalent to 
\begin{align}
&\nn \frac{1}{T_n}\sum_{i=1}^{\lf T_n\rf}E^{i-1}\left[\p_\eta m(\ep_i, \eta_0)\sum_{j\in A_i} \left(\frac{\p_\gam c_{j-1}}{c_{j-1}}\left[(\D_j J)^2-h \right]\right)\right]\cip \Sigma_{\eta,\gam},\\
&\nn \frac{1}{T_n}\sum_{i=1}^{\lf T_n\rf}E^{i-1}\left[\p_\eta m(\ep_i, \eta_0)\sum_{j\in A_i}\left(\frac{\p_\al a_{j-1}}{c_{j-1}}\D_j J\right)\right]  \cip \Sigma_{\eta,\al},
\end{align}
respectively.
This fact will be used later.
\end{Rem}

Let
\begin{align}
\hat{\Gam}_n &:=\diag(-\p_\gam^2\mbbh_n(\ges), -\p_\al^2\mbbh_n(\aes)), \nn\\
\hat{\mci}_n &:=\begin{pmatrix}\hat{\Gam}_n &O\\ - T_n^{-1} \sum_{i=1}^{\left\lf T_n\right\rf} \p_\eta\p_\ep m\left(\hat{\ep}_i,\ees\right) \hat{b}_i &
 -\p_\eta^2\mbbh_{3,n}(\ees)
\end{pmatrix}.
\nonumber
\end{align}
We also introduce the $(p+p_\eta)\times(p+p_\eta)$-matrix
\begin{equation*}
\hat{\Sig}_n=\begin{pmatrix}\hat{\Sig}_{\gam,n} & \hat{\Sig}_{\gam,\al,n} & \hat{\Sig}_{\eta,\gam,n}\\\hat{\Sig}_{\gam,\al,n}^\top & \hat{\Sig}_{\al,n} &\hat{\Sig}_{\eta,\al,n}\\ \hat{\Sig}_{\eta,\gam,n}^\top& \hat{\Sig}_{\eta,\al,n}^\top & \hat{\Sig}_{\eta,n}\end{pmatrix},
\end{equation*}
where the ingredients are defined as follows:
\begin{align*}
&\hat{\Sig}_{\gam,n}=\left(\frac{1}{n}\sumj \left(\frac{\p_\gam \hat{c}_{j-1}}{\hat c_{j-1}}\right)^{\otimes2}\right)\left(\frac{1}{h_n}\sumj  \left(\widehat{\D_j J}\right)^4\right),
\nn\\
&\hat{\Sig}_{\al,n}=\frac{1}{n}\sumj \left(\frac{\p_\al \hat{a}_{j-1}}{\hat c_{j-1}}\right)^{\otimes2},
\nn\\
& \hat{\Sig}_{\eta,n}=\frac{1}{T_n}\sum_{i=1}^{\lf T_n\rf}(\p_\eta m(\hat{\ep}_i, \ees))^{\otimes2},
\nn\\
& \hat{\Sig}_{\gam,\al,n}=\left(\frac{1}{n}\sumj \frac{\p_\gam \hat{c}_{j-1}\p_\al \hat{a}_{j-1}}{\hat {c}_{j-1}^2}\right)\left(\frac{1}{h_n}\sumj  \left(\widehat{\D_j J}\right)^3\right),\\
&\hat{\Sig}_{\eta,\gam,n}= \frac{1}{T_n}\sum_{i=1}^{\lf T_n\rf}\p_\eta m(\hat{\ep}_i, \ees)\sum_{j\in A_i} \left(\frac{\p_\gam \hat{c}_{j-1}}{\hat{c}_{j-1}}\left[\left(\widehat{\D_j J}\right)^2-h \right]\right), \\
&\hat{\Sig}_{\eta,\al,n}=\frac{1}{T_n}\sum_{i=1}^{\lf T_n\rf}\p_\eta m(\hat{\ep}_i, \ees)\sum_{j\in A_i}\left(\frac{\p_\al \hat{a}_{j-1}}{\hat{c}_{j-1}}\widehat{\D_j J}\right).
\end{align*}
%
Now we are ready to state the main result.

\begin{Thm}
\label{hm:thm_u-AN}
Under Assumptions \ref{Moments} to \ref{Identifiability}, Assumption \ref{estimating}, and Assumption \ref{yu:cov}, we have
\begin{equation}
\nn\hat{\Sig}_n^{-1/2}\hat{\mci}_n[\hat{u}_n]\cil N_{p+p_\eta}\left(0, I_{p+p_\eta}\right).
\end{equation}
\end{Thm}

By Theorem \ref{hm:thm_u-AN}, we have
\begin{equation}
T_n 
(\tes-\tz,\ees-\eta_0)\,
\hat{\mci}_n^{\top} \hat{\Sig}_n^{-1} \hat{\mci}_n
(\tes-\tz,\ees-\eta_0)^{\top}
\cil \chi^{2}(p+p_\eta),
\nonumber
\end{equation}
based on which we can construct an approximate confidence set for $(\tz,\eta_0)$, and also perform a Wald-type test.
Also trivially, we can recover the asymptotic distribution of $\tes$:
\begin{equation}
T_n 
(\tes-\tz)\,
\hat{\Gam}_n
\begin{pmatrix}\hat{\Sig}_{\gam,n} & \hat{\Sig}_{\gam,\al,n} \\
\hat{\Sig}_{\gam,\al,n}^\top & \hat{\Sig}_{\al,n}
\end{pmatrix}^{-1}
\hat{\Gam}_n
(\tes-\tz)^{\top} \cil \chi^{2}(p).
\nonumber
\end{equation}

It is difficult to obtain $\Sig_{\eta,\gam}$ and $\Sig_{\eta,\al}$ in explicit easy-to-handle forms even if the coefficients and $m$ are simple.
However, by an application of Cauchy-Schwartz inequality and the estimates in the proof of Theorem 3.8,
we can observe that at least, the left-hand-sides in \eqref{yu:cr1} and \eqref{yu:cr2} are tight.
Moreover, we can formally write their limit by means of the representation theorem (\cite[Proposition 3]{Lok04}): there exists a predictable process $s\mapsto \tilde{\xi}_{\eta,n}(s,z)$ such that
\begin{equation}
\sum_{i=1}^{\lf T_n\rf}\p_\eta m(\ep_i,\eta_0)=\int_{0}^{\lf T_n\rf}\int \tilde{\xi}_{\eta,n}(s,z)\tilde{N}(ds,dz).
\end{equation}
From It\^{o}'s formula and some calculations, we have
\begin{align*}
&\sum_{j\in A_i}\frac{\p_\al a_{j-1}}{c_{j-1}}\D_j J\\
&=\int_{\left\lf \frac{i-1}{h }\right\rf h }^{\left\lf \frac{i}{h }\right\rf h }\sum_{j\in A_i} \left(\chi_j(s)\frac{\p_\al a_{j-1}}{c_{j-1}}\right) dJ_s=\int_{\left\lf \frac{i-1}{h }\right\rf h }^{\left\lf \frac{i}{h }\right\rf h }\int\sum_{j\in A_i} \left(\chi_j(s)\frac{\p_\al a_{j-1}}{c_{j-1}}\right) z\tilde{N}(ds,dz)\\
&=\int_{\left\lf \frac{i-1}{h }\right\rf h }^{\left\lf \frac{i}{h }\right\rf h }\int\frac{\p_\al a_{s-}}{c_{s-}} z\tilde{N}(ds,dz),\\
&\sum_{j\in A_i}\frac{\p_\gam c_{j-1}}{c_{j-1}}\left[(\D_j J)^2-h \right]\\
&=\int_{\left\lf \frac{i-1}{h }\right\rf h }^{\left\lf \frac{i}{h }\right\rf h }\sum_{j\in A_i} \left(\chi_j(s)\frac{\p_\gam c_{j-1}}{c_{j-1}}(J_{s-} - J_{t_{j-1}})\right) dJ_s+\int_{\left\lf \frac{i-1}{h }\right\rf h }^{\left\lf \frac{i}{h }\right\rf h }\int \sum_{j\in A_i} \left(\chi_j(s)\frac{\p_\gam c_{j-1}}{c_{j-1}}\right)z^2\tilde{N}(ds,dz)\\
&=\int_{\left\lf \frac{i-1}{h }\right\rf h }^{\left\lf \frac{i}{h }\right\rf h }\int \frac{\p_\gam c_{s-}}{c_{s-}}z^2\tilde{N}(ds,dz)+O_p(h_n).
\end{align*}
To sum up, we obtain the following expression:
\begin{align*}
&\hat{\mci}_n[\hat{u}_n]=\frac{1}{\sqrt{T_n}}\int_{0}^{\lf T_n\rf}\int \begin{pmatrix}\frac{\p_\gam c_{s-}}{c_{s-}}z^2\\\frac{\p_\al a_{s-}}{c_{s-}} z\\\tilde{\xi}_{\eta,n}(s,z)\end{pmatrix}\tilde{N}(ds,dz)+o_p(1).
\end{align*}
By applying the central limit theorem for the stochastic integral with respect to a Poisson random measure (cf. \cite[Lemma A.2]{Ueh19}), the isometry property of the stochastic integral yield that under suitable moment and regularity conditions,
\begin{align*}
\hat{\mci}_n[\hat{u}_n]\overset{\mcl}\rightarrow N_{p+p_\eta}\left(0,  \begin{pmatrix}\Sig_\gam & \Sig_{\gam,\al} & \Sig_{\eta,\gam}\\\Sig_{\gam,\al}^\top & \Sig_{\al} &\Sig_{\eta,\al}\\ \Sig_{\eta,\gam}^\top& \Sig_{\eta,\al}^\top & E\left[\left(\p_\eta m(J_1,\eta_0)\right)^{\otimes2}\right] \end{pmatrix}\right),
\end{align*}
where $\Sig_{\eta,\gam}$ and $\Sig_{\eta,\al}$ are given by the limits in probability:
\begin{align*}
&\frac{1}{T_n}\int_{0}^{\lf T_n\rf}\int E\left[\frac{\p_\gam c_{s-}}{c_{s-}}\tilde{\xi}_{\eta,n}(s,z)\right]z^2\nu_0(dz)ds
\cip \Sig_{\eta,\gam}, \\
&\frac{1}{T_n}\int_{0}^{\lf T_n\rf}\int E\left[\frac{\p_\al a_{s-}}{c_{s-}}\tilde{\xi}_{\eta,n}(s,z)\right]z\nu_0(dz)ds
\cip \Sig_{\eta,\al},
\end{align*}
and the other ingredients are the same as our previous works (cf. Theorem \ref{hm:thm_old}).
However, since the explicit form of $\tilde{\xi}_{\eta,n}(s,z)$ cannot be obtained in general, it is difficult to check the above convergence.

Finally, we would like to add that Theorem \ref{hm:thm_u-AN} and the resulting Wald-type test are valid without Assumption \ref{yu:cov} if the minimum eigenvalue of $\hat{\Sig}$ is positive uniformly in $n$.
Such a condition for eigenvalues is often assumed in the context of (non-)linear regression.

\begin{proof}[Proof of Theorem \ref{hm:thm_u-AN}]
By the Cram\'{e}r-Wold device, we may and do assume that $p_\gam=p_\al=p_\eta=1$ without loss of generality.
It is straightforward to deduce that $\hat{\Sig}_n\cip \Sig$ from Theorem \ref{hm:thm_old}, Theorem \ref{yu:se}, Lemma \ref{yu:gap}, and the estimates we have seen in the previous proofs.
Hence, by means of Slutsky's theorem, it suffices to show that
\begin{equation}\label{yu:jointan}
\hat{\mci}_n[\hat{u}_n]\cil N_{p+p_\eta}\left(0, \Sig\right).
\end{equation}
From \cite[Proof of Theorem 3.4]{MasUeh17-2}, we have
\begin{align}
    \nn\hat{\Gam}_n[\hat{u}_{\theta,n}]=\frac{1}{\sqrt{T_n}}\sum_{i=1}^{\lf T_n\rf}\sum_{j\in A_i}\begin{pmatrix}\frac{\p_\gam c_{j-1}}{c_{j-1}}\left[(\D_j J)^2-h \right]\\\frac{\p_\al a_{j-1}}{c_{j-1}}\D_j J\end{pmatrix}+o_p(1).
\end{align}
This together with \eqref{yu:score} and the definition of $\hat{\mci}_n$ leads to
\begin{equation}
\label{hm:asymp.linear.rep}
\hat{\mci}_n[\hat{u}_n]=\frac{1}{\sqrt{T_n}}\sum_{i=1}^{\lf T_n\rf}\begin{pmatrix}\xi_{\gam,i}\\\xi_{\al,i}\\\xi_{\eta,i}\end{pmatrix}+o_p(1),
\end{equation}
where
\begin{align}
&\nn \xi_{\gam,i}:=\sum_{j\in A_i}\frac{\p_\gam c_{j-1}}{c_{j-1}}\left[(\D_j J)^2-h \right], 
\quad \xi_{\al,i}:=\sum_{j\in A_i}\frac{\p_\al a_{j-1}}{c_{j-1}}\D_j J, 
\quad \xi_{\eta,i}:=\p_\eta m\left(\ep_i,\eta_0\right).
\end{align}
By \eqref{yu:cr1}, \eqref{yu:cr2}, and the arguments in the proof of Corollary \ref{yu:clteta}, the martingale central limit theorem concludes \eqref{yu:jointan} if we have the following convergences: 
\begin{align}
&\label{yu:an1}\left|\frac{1}{\sqrt{T_n}}\sum_{i=1}^{\lf T_n\rf}E^{i-1}\left[\xi_{\gam,i}\right]\right|+\left|\frac{1}{\sqrt{T_n}}\sum_{i=1}^{\lf T_n\rf}E^{i-1}\left[\xi_{\al,i}\right]\right|\cip0,\\
&\label{yu:an2}\frac{1}{T_n}\sum_{i=1}^{\lf T_n\rf}E^{i-1}\left[\xi_{\gam,i}^2\right]\cip \Sig_\gam,\\
&\label{yu:an3}\frac{1}{T_n}\sum_{i=1}^{\lf T_n\rf}E^{i-1}\left[\xi_{\gam,i}\xi_{\al,i}\right]\cip \Sig_{\gam,\al},\\
&\label{yu:an4}\frac{1}{T_n}\sum_{i=1}^{\lf T_n\rf}E^{i-1}\left[\xi_{\al,i}^2\right]\cip \Sig_\al,\\
&\label{yu:an5}\frac{1}{T_n^2}\sum_{i=1}^{\lf T_n\rf}E^{i-1}\left[|\xi_{\gam,i}|^4+|\xi_{\al,i}|^4+|\xi_{\eta,i}|^4\right]\cip0.
\end{align}
Trivially $\{(\xi_{\gam,i},\xi_{\al,i},\xi_{\eta,i})\}_{i\le \lf T_n\rf}$ forms a martingale difference array with respect to $(\mcf_i)$, since we have $\mcf_{i-1} \subset \mcf_{t_j}$ for each $j\in A_i$; this immediately ensures \eqref{yu:an1}.
By the arguments in Remark \ref{yu:cond}, we can replace \eqref{yu:an2}, \eqref{yu:an3}, and \eqref{yu:an4} by
\begin{align*}
&
\frac{1}{T_n}\sumj\left(\frac{\p_\gam c_{j-1}}{c_{j-1}}\left[(\D_j J)^2-h \right]\right)^2\cip \Sig_\gam,\\
&
\frac{1}{T_n}\sumj\left(\frac{\p_\gam c_{j-1}}{c_{j-1}}\left[(\D_j J)^2-h \right]\right)\left(\frac{\p_\al a_{j-1}}{c_{j-1}}\D_j J\right)\cip \Sig_{\gam,\al},\\
&
\frac{1}{T_n}\sumj\left(\frac{\p_\al a_{j-1}}{c_{j-1}}\D_j J\right)^2\cip \Sig_\al,
\end{align*}
respectively.
Noting that $E[|J_{h_n}|^q]=O(h_n)$ for any $q\geq2$ under Assumption \ref{Moments}, we can deduce the last three convergences from \cite[Lemma 9]{GenJac93} and the ergodic theorem.

It remains to show \eqref{yu:an5}.
It follows from It\^{o}'s formula and Assumption \ref{Moments} that for any $j\in\{1,\dots, n\}$,
\begin{equation*}
(\D_j J)^2-h = 2\int_{t_{j-1}}^{t_j} (J_{s-} - J_{t_{j-1}}) dJ_s + \int_{t_{j-1}}^{t_j}\int z^2 \tilde{N}(ds,dz),
\end{equation*}
where $\tilde{N}(ds,dz)$ is the compensated Poisson random measure of $J$; recall that we are assuming that $\E[(\D_j J)^2]=(\sig^2 + \int z^2\nu(dz))h=h$.
Then, we can rewrite $\xi_{\gam,i}$ and $\xi_{\al,i}$ as
\begin{align}
\nn &\xi_{\gam,i}=\int_{\left\lf \frac{i-1}{h }\right\rf h }^{\left\lf \frac{i}{h }\right\rf h }\sum_{j\in A_i} \left(\chi_j(s)\frac{\p_\gam c_{j-1}}{c_{j-1}}(J_{s-} - J_{t_{j-1}})\right) dJ_s+\int_{\left\lf \frac{i-1}{h }\right\rf h }^{\left\lf \frac{i}{h }\right\rf h }\int \sum_{j\in A_i} \left(\chi_j(s)\frac{\p_\gam c_{j-1}}{c_{j-1}}\right)z^2\tilde{N}(ds,dz),\\
\nn &\xi_{\al,i}=\int_{\left\lf \frac{i-1}{h }\right\rf h }^{\left\lf \frac{i}{h }\right\rf h }\sum_{j\in A_i} \left(\chi_j(s)\frac{\p_\al a_{j-1}}{c_{j-1}}\right) dJ_s=\int_{\left\lf \frac{i-1}{h }\right\rf h }^{\left\lf \frac{i}{h }\right\rf h }\int\sum_{j\in A_i} \left(\chi_j(s)\frac{\p_\al a_{j-1}}{c_{j-1}}\right) z\tilde{N}(ds,dz).
\end{align}
Taking a similar route to the estimate \eqref{yu:burk}, we have
\begin{align}
\nn\frac{1}{T_n^2}\sum_{i=1}^{\lf T_n\rf}E^{i-1}\left[|\xi_{\gam,i}|^4+|\xi_{\al,i}|^4\right]\lesssim \frac{1}{T_n^2}\sum_{i=1}^{\lf T_n\rf} (1+|X_{i-1}|^C)=O_p(T_n^{-1})=o_p(1).
\end{align}
Since $T_n^{-2}\sum_{i=1}^{\lf T_n\rf}E^{i-1}\left[|\xi_{\eta,i}|^4\right]\cip0$ under Assumptions \ref{Moments} and \ref{estimating}, we obtain the desired result.
\end{proof}


\subsection{Further remarks}

\subsubsection{Dimension of the processes}

For the asymptotics of the GQMLE, we could consider multivariate $X$ without any essential change \cite{MasUeh17-2};
we will conduct related simulations in Section \ref{hm:sec_multivariate.SDE.sim}.
Moreover, the estimator of $(\al,\gam)$ may not be necessarily the GQMLE and could be any measurable mappings $\tes=\tes(\bm{X}_n)$ for which we have an asymptotically linear representation as in \eqref{hm:asymp.linear.rep}.

\subsubsection{Model selection for $\mcl(J_1)$}

Residual based on information criterion (IC) formulation after estimation should be possible (both AIC and BIC types).
We can infer the structure of $J$ as in the i.i.d case, yet should be careful in making possibly necessary corrections stemming from the stochastic expansion \eqref{yu:se-1}.
For example, for the AIC statistics to be theoretically in effect, among other conditions it is required that the random sequence $(\hat{u}_{\eta,n})_n$ is $L^{2+\del}(\pr)$-bounded for some $\del>0$. It could be verified by means of the uniform tail-probability estimate for $(\hat{u}_{\eta,n})_n$ through the random function $\mbbh_{3,n}(\eta)$; indeed, we could make use of the same machinery to deduce Theorem \ref{hm:thm_old}(2).

\subsubsection{Setting of noise inference}

Although we have set a finite-dimensional $\eta$ above, we could consider infinite-dimensional $\eta$, most generally $\mcl(J_1)$ itself: once $\{\hat{\ep}_j\}$ has been constructed, it is also possible to take into account conventional nonparametric procedures, such as the kernel density estimation, and also goodness-of-fit tests; see Section \ref{hm:sec_real.data} for an illustration.


\section{Implementation} \label{sec:Implementation}

In this section, we discuss the new classes and the new methods in the \yuima  \proglang{R} package that gives us the possibility to deal with an SDE driven by a L\'evy process completely specified by the user. To construct an object of \yuima  class, that is a mathematical description of an SDE driven by  a pure L\'evy jump process, three steps are necessary:
\begin{enumerate}
\item Definition of an object that contains all the information about the structure of the pure L\'evy jump. In this step, the user can specify a random number generator, a density function, a cumulative distribution function, a quantile function, a characteristic function, and the number of components for the underlying L\'evy process.  
\item Definition of the structure of the SDE where the driving noise is determined from the object constructed in Step 1.
\item Construction of an object that belongs to the \yuima class whose slots are reported in Figure \ref{Fig:slots}. The slot \texttt{model} is filled with the object built in Step 2.  This new object can be used to simulate a sample path by overwriting the slot \texttt{sampling} with the structure of the time grid. Alternatively, we can use this object to estimate the SDE defined in Step 2. In this case, we can store the observed data in the slot \texttt{data}.  
\end{enumerate}
\begin{figure}[!h]
\centering
\begin{tikzpicture} 
\tikzstyle{comment}=[rectangle, draw=black, rounded corners, fill=gray!10, text=black, drop shadow, text width=8.2cm]
    \node (PressureInstants) [comment, text justified]
        {\ \ \ \ \ \ \ \ \ \ \ \ \ \ \ \ \ \ \ \ \ \ \footnotesize{Object of \yuima class}
				
						\footnotesize{\texttt{Slots}:}\newline
						$\begin{array}{|ll|}
						\hline 
						\emph{{\footnotesize{@}}\footnotesize{data}}:&\text{Data that can be either real or simulated}\\ 
						\emph{{\footnotesize{@}}\footnotesize{model}}:&\text{Mathematical description of the model}\\
						\emph{{\footnotesize{@}}\footnotesize{sampling}}:&\text{Structure of the time grid}\\
						\emph{{\footnotesize{@}}\footnotesize{characteristic}}:&\text{Additional info}\\
						\emph{{\footnotesize{@}}\footnotesize{functional}}:&\text{Functional of SDE}\\
						\hline
						\end{array}$
					};
\end{tikzpicture}
\caption{The structure of an object that belongs to the \yuima class.}
\label{Fig:slots}
\end{figure}
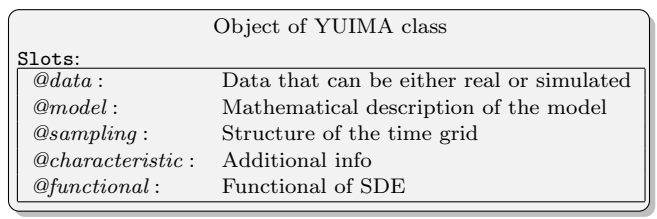

\subsection{yuima.law: A New Class for a Mathematical Description of the L\'evy process}

In this section, we describe the structure of a \texttt{yuima.law-object} and its constructor \texttt{setLaw}. The main advantage of this new class is the possibility of connecting \yuima with any CRAN package that provides functions for a specific random variable. Figure \ref{Fig:slotslaw} reports the slots that constitute an object of \texttt{yuima.law} class.

\begin{figure}[!h]
\centering
\begin{tikzpicture} 
\tikzstyle{comment}=[rectangle, draw=black, rounded corners, fill=gray!10, text=black, drop shadow, text width=7.5cm]
    \node (PressureInstants) [comment, text justified]
        {\ \ \ \ \ \ \ \ \ \ \ \ \ \ \ \ \ \footnotesize{Object of \texttt{yuima.law} class}
				
						\footnotesize{\texttt{Slots}:}\newline
						$\begin{array}{|ll|}
						\hline 
						\emph{{\footnotesize{@}}\footnotesize{rng}}:&\text{Random number generator}\\ 
						\emph{{\footnotesize{@}}\footnotesize{density}}:&\text{Density function}\\
						\emph{{\footnotesize{@}}\footnotesize{cdf}}:&\text{Cumulative distribution function}\\
						\emph{{\footnotesize{@}}\footnotesize{quantile}}:&\text{Quantile function}\\
						\emph{{\footnotesize{@}}\footnotesize{characteristic}}:&\text{Characteristic function}\\
						\emph{{\footnotesize{@}}\footnotesize{param.measure}}:&\text{Parameter labels}\\
						\emph{{\footnotesize{@}}\footnotesize{time.var}}:&\text{Time label}\\
						\emph{{\footnotesize{@}}\footnotesize{dim}}:&\text{Dimension of the random variable}\\
						\hline
						\end{array}$
					};
\end{tikzpicture}
\caption{The structure of an object that belongs to the \texttt{yuima.law} class.\label{Fig:slotslaw}}
\end{figure}
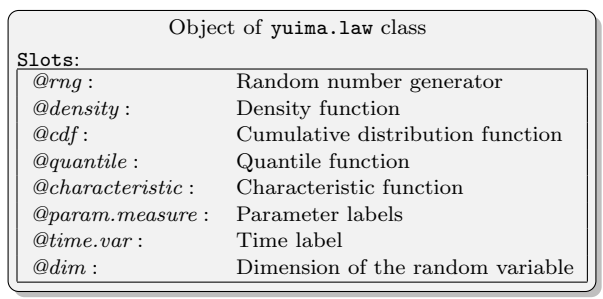

The first five slots contain \textsf{R} user-defined functions. In particular, the first two slots contain the random number generator and the density function respectively. Although it is not necessary to specify these functions to construct an object of \texttt{yuima.law} class, the definition of a random number generator is necessary to run the \yuima \texttt{simulate} method while the density function is used internally by the \yuima \texttt{qmleLevy} method. The template of these two functions is listed below:
\begin{CodeChunk}
\begin{CodeInput}
# User specified random number generator
R> user.rng <- function(n, eta, t){
+    ... ... ... # Body of the function 
+  }
# User specified density function
R> user.density <- function(x, eta, t){
+    ... ... ... # Body of the function 
+  }
\end{CodeInput}
\end{CodeChunk} 
where the input \texttt{eta} is a \texttt{vector} containing the names of the L\'evy noise parameters and the input \texttt{t} refers to the label of the time variable. 

An object of \texttt{yuima.law} class is built using \texttt{setLaw} constructor.
\begin{CodeChunk}
\begin{CodeInput}
R> setLaw(rng = function(n, ...){ NULL }, density = function(x, ...){ NULL }, 
+    cdf = function(q, ...){ NULL }, quant = function(p, ...){ NULL },
+    characteristic = function(u, ...){ NULL }, time.var = "t",
+    dim = NA )
\end{CodeInput}
\end{CodeChunk}
The first five inputs in the function fill the corresponding slots in the \texttt{yuima.law} object. Figure \ref{TIKZFIG} describes the steps required for the construction of the \texttt{yuima.law-object}. 
\pgfdeclarelayer{background}
\pgfdeclarelayer{foreground}
\pgfsetlayers{background,main,foreground}
\tikzstyle{sensor}=[draw, fill=gray!10, text width=7.5em, 
    text centered, minimum height=3em,drop shadow]
\tikzstyle{sensor1}=[draw, fill=gray!10, text width=13em, 
    text centered, minimum height=3em,drop shadow]
\def\blockdist{2.3}
\def\edgedist{2.5}
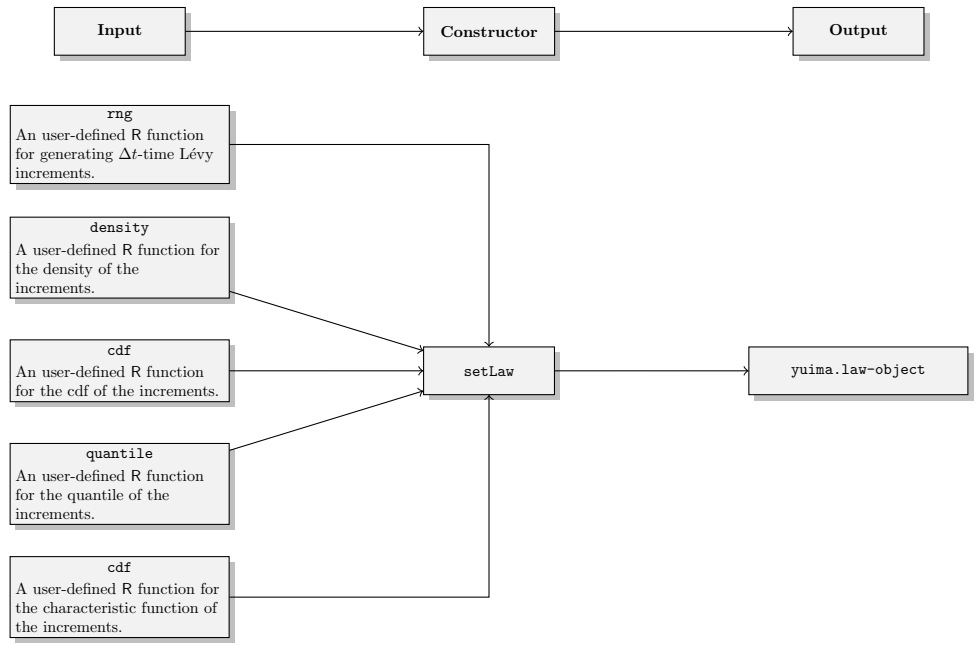
\begin{figure}[!h]
\centering
\begin{tikzpicture}[scale=0.60, transform shape]
    \node (pos1) [sensor]  {\textbf{Output}};
		\path (pos1)+(-8.1,0) node (cons) [sensor]{\textbf{Constructor}};
    \path (pos1)+(-16.2, 0) node (Inp) [sensor] {\textbf{Input}};
		\path (pos1)+(-16.2, -2.5) node (rng) [sensor1] {
		\texttt{rng} 
		\begin{flushleft}
		An user-defined \proglang{R} function for generating $\Delta t$-time L\'evy increments.
		\end{flushleft}
		};
		\path (rng)+(0, -2.5) node (density) [sensor1] {
		\code{density} 
		\begin{flushleft}
		A user-defined \proglang{R} function for the density of the increments.
		\end{flushleft}
		};
		\path (density)+(0, -2.5) node (cdf) [sensor1] {
		\code{cdf} 
		\begin{flushleft}
		An user-defined \proglang{R} function for the cdf of the increments.
		\end{flushleft}
		};
		\path (cdf)+(0, -2.5) node (quantile) [sensor1] {
		\code{quantile} 
		\begin{flushleft}
		An user-defined \proglang{R} function for the quantile of the increments.
		\end{flushleft}
		};
		\path (quantile)+(0, -2.5) node (char) [sensor1] {
		\code{cdf} 
		\begin{flushleft}
		A user-defined \proglang{R} function for the characteristic function of the increments.
		\end{flushleft}
		};
		\path (cdf)+(8.1, 0) node (sl) [sensor] {
		\code{setLaw} 
		};
		\path (sl)+(8.1, 0) node (yl) [sensor1] {
		\code{yuima.law-object} 
		};
		\draw[->] (Inp) edge (cons);
		\draw[->] (cons) edge (pos1);
		\draw[->] (sl) edge (yl);
		\draw[->] (cdf) edge (sl);
		\draw[->,  to path={-| (\tikztotarget)}] (rng) edge (sl);
		\draw[->] (density) edge (sl);
		\draw[->] (quantile) edge (sl);
		\draw[->,  to path={-| (\tikztotarget)}] (char) edge (sl);
\end{tikzpicture}
\caption{Procedure for the construction of an object that belongs to the \code{yuima.law} class. \label{TIKZFIG}}
\end{figure}

After the construction of an object that belongs to the \code{yiuma.law} class, by using the standard constructor \code{setModel} where an \code{yuima.law} object is passed to \code{setModel} through the argument \code{measure}, the user can specify completely a SDE driven by a pure L\'evy jump as shown in the following command line:
\begin{CodeChunk}
\begin{CodeInput}
R> setModel(drift = "User.Defined_drift", jump.coeff = "User.Defined_jump.coef",        
+    measure.type = "code", measure = list(df = User.Defined_yuima.law))
\end{CodeInput}
\end{CodeChunk}
We remark that an object of \code{yuima.law} class can be also used to specify the L\'evy noise in the Continuous Time ARMA model \cite{iacus2015implementation} and in the COGARCH process \cite{iacus2017cogarch,iacus2018discrete}. In the first case, the model is built using the constructor \code{setCarma}:
\begin{CodeChunk}
\begin{CodeInput}
R> setCarma(p, q, measure.type = "code", measure = list(df = User.Defined_yuima.law))
\end{CodeInput}
\end{CodeChunk} 
where $p$ and $q$ are two integers indicating the order of the autoregressive and the moving average parameters. The COGARCH(p,q) process can be defined in \yuima using the function \code{setCogarch} as follows:
\begin{CodeChunk}
\begin{CodeInput}
R> setCogarch(p, q, measure.type = "code", measure = list(df = User.Defined_yuima.law))
\end{CodeInput}
\end{CodeChunk}

\begin{figure}[!h]
\centering
\begin{tikzpicture}[scale=0.60, transform shape]
    \node (pos1) [sensor]  {\textbf{Outputs}};
		\path (pos1)+(-8.1, 0) node (cons) [sensor]{\textbf{Constructors}};
		\path (pos1)+(-16.2, 0) node (Inp) [sensor]  {\textbf{Inputs}};
		\path (cons)+(0, -2.1) node (setC) [sensor] {\code{setCarma}};
		\path (setC)+(0, -2.1) node (setM) [sensor] {\code{setModel}};
		\path (setM)+(0, -2.1) node (setCog) [sensor] {\code{setCogarch}};
		\path (setM)+(-8.1, 0) node (Inp1) [sensor] {\code{yuima.law} and additional inputs};
		\draw[-> ] (Inp1) edge (setC); 
		\draw[-> ] (Inp1) edge (setM);
		\draw[-> ] (Inp1) edge (setCog);
		\path (setC)+(8.1,0) node (OupC) [sensor] {\code{yuima.carma}};
		\draw[-> ] (setC) edge (OupC);
		\path (setM)+(8.1,0) node (OupM) [sensor] {\code{yuima.model}};
		\draw[-> ] (setM) edge (OupM);
		\path (setCog)+(8.1,0) node (OupCog) [sensor] {\code{yuima.model}};
		\draw[-> ] (setCog) edge (OupCog);
\end{tikzpicture}
\caption{Procedure for the definition of the model where the noise is defined by an object of \code{yuima.law} class. \label{TIKZFIG1}}
\end{figure}
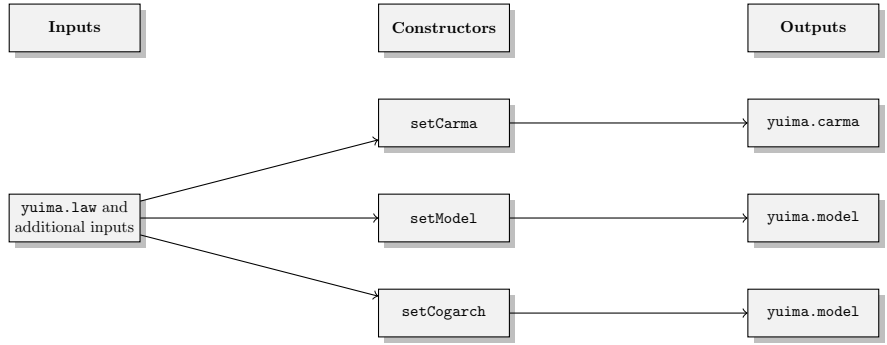
Figure \ref{TIKZFIG1} shows how to use an object of \code{yuima.law} class in the definition of models that can be constructed using  YUIMA.

\subsection{yuima.qmleLevy.incr: Estimation of an SDE driven by a L\'evy pure jump process in yuima}
In this section, we discuss how to estimate an SDE driven by a L\'evy pure jump process in YUIMA. In particular, we describe the features of the new class \code{yuima.qmleLevy.incr} and explain the usage of the new method \code{qmleLevy}. The \code{yuima.qmleLevy.incr} class is the extension of the classical \code{yuima.qmle} class because we have additional slots associated with the filtered L\'evy increments obtained using the procedure described in Section \ref{yu:main}. As a child class, \code{yuima.qmleLevy.incr} class inherits all the \yuima methods developed for the \code{yuima.qmle} class. Figure \ref{Fig:slotsyuimaqmle} reports the new slots. The most relevant for our study is the slot \code{Incr.Lev} where we can find the estimated L\'evy increments. 

\begin{figure}[!h]
\centering
\begin{tikzpicture} 
\tikzstyle{comment}=[rectangle, draw=black, rounded corners, fill=gray!10, text=black, drop shadow, text width=9.3cm]
    \node (PressureInstants) [comment, text justified]
        {\ \ \ \ \ \ \ \ \ \ \ \ \ \ \ \ \footnotesize{Object of \textbf{\code{yuima.qmleLevy.incr}} class} \newline
						$\begin{array}{|ll|}
						\hline 
						\emph{{\footnotesize{@}}\footnotesize{Incr.Lev}}:&\text{Estimated $\Delta t$ or unit-time L\'evy increments}\\ 
						\emph{{\footnotesize{@}}\footnotesize{logL.Incr}}:&\text{Log-likelihood of the estimated L\'evy increments}\\
						\emph{{\footnotesize{@}}\footnotesize{Levydetails}}:&\text{Additional information on the internal optimization}\\
						\emph{{\footnotesize{@}}\footnotesize{Data}}:&\text{Observed data}\\
						\emph{{\footnotesize{@}}\footnotesize{...}}:&\text{Slots inherited from \code{yuima.qmle} class}\\
						\hline
						\end{array}$
					};
\end{tikzpicture}
\caption{Main slots of an object that belongs to the \code{yuima.qmleLevy.incr} class.\label{Fig:slotsyuimaqmle}}
\end{figure}

An object of \code{yuima.qmleLevy.incr} class can not be directly constructed by the user but it is a possible output of the function \code{qmleLevy} that performs the estimation approach discussed in Section \ref{yu:notation and assumption}. The syntax of this function is as follows:
\begin{CodeChunk}
\begin{CodeInput}
R> qmleLevy(yuima, start, lower, upper, joint = FALSE, third = FALSE,  
+    Est.Incr = "NoIncr", aggregation = TRUE)
\end{CodeInput}
\end{CodeChunk}
The first argument is an object of \code{yuima} class where the slot \code{data} contains the observed dataset, while the slot \code{model} is a mathematical description of the SDE driven by the pure L\'evy jump process. The arguments \code{start}, \code{lower} and \code{upper} are used in the optimization routine to identify the initial guesses and box-constraints. The arguments \code{joint} and \code{third} are technical arguments related to the procedure of the GQMLE; we refer to \cite{iacus2018simulation} for a specific documentation of their meaning. The most important arguments for the estimation of the L\'evy increments are \code{Est.Incr} and \code{aggregation}. The argument \code{Est.Incr} assumes three values: \code{NoIncr}, \code{Incr} and \code{IncrPar}. In the first case, the function returns an object of \code{yuima.qmle} class that contains only the SDE parameters. The function \code{qmleLevy} internally runs only the GQMLE procedure. Setting \code{Est.Incr = "Incr"} or \code{Est.Incr = "IncrPar"}, \code{qmleLevy} returns an object of \code{yuima.qmleLevy.incr} class. In the first case the object contains the estimated increments while, in the second case, we also obtain the estimated parameters of the L\'evy measure. The last argument \code{aggregation} is a logical variable.
If \code{aggregation = TRUE}, the estimated L\'evy increments are \eqref{hm:add2} associated to the unit-time intervals while, if \code{aggregation = FALSE}, the function returns the $\Delta t$-time L\'evy increments, see \eqref{hm:add1}.

\begin{figure}[!h]
\centering
\begin{tikzpicture}[scale=0.60, transform shape]
\node (pos1) [sensor]  {\code{yuima.model}};
\path (pos1)+(12, 0) node (pos2) [sensor]{\code{yuima.data}};
\path (pos1)+(6, -2.1) node (pos3) [sensor]{\code{setYuima}};
\draw[-> ] (pos1) edge (pos3);
\draw[-> ] (pos2) edge (pos3);
\path (pos3)+(0, -2.1) node (Out1) [sensor]{\code{yuima}};
\draw[-> ] (pos3) edge (Out1);
\path (Out1)+(6, 0) node (Inp1) [sensor] {\code{start}, \code{lower}, \code{upper}, \code{...}};
\path (Out1)+(0, -2.1) node (cons) [sensor] {\code{qmleLevy}};
\draw[-> ] (Out1) edge (cons);
\draw[->,  to path={|- (\tikztotarget)}] (Inp1) edge (cons);
\path (cons)+(-6, -2.1) node (1Out) [sensor] {\code{yuima.qmle}};
\path (cons)+(6, -2.1) node (2Out) [sensor1] {\code{yuima.qmleLevy.incr}};
\draw[-> ] (cons) edge (1Out);
\draw[-> ] (cons) edge (2Out);
\end{tikzpicture}
\caption{Procedure for the estimation of a SDE driven by a L\'evy process in YUIMA.\label{Fig:Const}}
\end{figure}
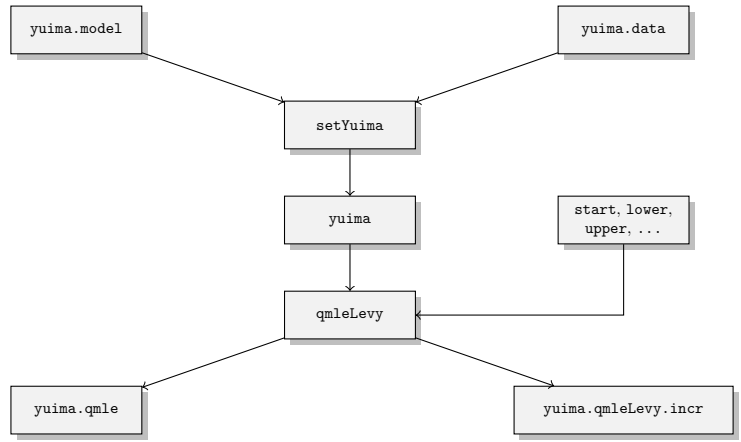
Figure \ref{Fig:Const} shows all the steps for the estimation of an SDE driven by a user-defined L\'evy process using the real data. As remarked at beginning of  Section \ref{sec:Implementation}, Figure \ref{Fig:Const} remarks the preliminary step for the construction of an object that belongs to \code{yuima} class through the constructor \code{setYuima}.   

\section{Numerical Examples} \label{sec:Numerical Examples}

\subsection{Univariate L\'evy SDE model}
In this section, we show how to use \yuima in the simulation and estimation of an univariate SDE driven by a pure jump L\'evy process defined by the user through an object of \code{yuima.law} class. The model, that we consider, is defined by the following SDE:
\begin{equation}
dX_t=\alpha_1\left(\alpha_2-X_t\right)dt + \gamma \mbox{d}J_t
\label{eq:OUsymVG}
\end{equation}
where $\alpha_2$ is a real parameter while $\alpha_1$ and $\gamma$ are positive parameters, and where $\left\{J_t\right\}_{t\geq0}$ is a symmetric Variance Gamma process with parameter $\eta > 0$. 

In this example we use an object of \code{yuima.law} class to construct a link between \yuima and the \texttt{VarianceGamma} package \cite{VGpackage} available in CRAN. We use two functions available in the package \texttt{VaranceGamma} respectively \code{rvg} for the random number generation and \code{dvg} to construct the density function. The parametrization in \texttt{VarianceGamma} package was introduced in \cite{Madan91} where the symmetric Variance Gamma $J$ random variable is defined as a normal variance mean mixture with a gamma subordinator.
Specifically, we set
\[
\phi_{J_1}\left(u\right)=\left(1+\sigma^2\nu \frac{u^2}{2}\right)^{-\frac{1}{\nu}}
\]
for the characteristic function of $J_1$.
Setting $\nu=\frac{1}{\eta\Delta t}$ and $\sigma=\sqrt{\Delta t}$, we identify the distribution of the increments $J_{t}-J_{t-\Delta t}$ for the symmetric Variance Gamma L\'evy process used in \eqref{eq:OUsymVG}. 

Following the structure presented in Section \ref{sec:Implementation}, we define an object of \code{yuima.law} class that contains all the information on the underlying process $\left\{J_t\right\}_{t\geq0}$. We run all examples using version \texttt{yuima.1.15.4} available on \proglang{R}-Forge. 
\begin{CodeChunk}
\begin{CodeInput}
R> library(VarianceGamma)

#### Definition of a yuima.law object ####

R> myrng <- function(n, eta, t){
+    rvg(n, vgC = 0, sigma = sqrt(t), theta = 0, nu = 1/(eta*t))
+  }

R> mydens <- function(x, eta, t){
+    dvg(x, vgC = 0, sigma = sqrt(t), theta = 0, nu = 1/(eta*t))
+  }

R> mylaw <- setLaw(rng = myrng, density = mydens, dim = 1)

R> class(mylaw)
\end{CodeInput}
\begin{CodeOutput}
[1] "yuima.law"
attr(,"package")
[1] "yuima"
\end{CodeOutput}

\begin{CodeInput}
R> slotNames(mylaw)
\end{CodeInput}
\begin{CodeOutput}
[1] "rng"            "density"        "cdf"            "quantile"       
[5] "characteristic" "param.measure"  "time.var"       "dim"      
\end{CodeOutput}
\end{CodeChunk}

Using the constructor \code{setLaw} we are able to build an object of \code{yuima.law} class where the first two slots contain the random number generator (\code{myrng}) and the density function (\code{mydens}) that we will use for the simulation and the estimation of the distribution of $J$ in the model \eqref{eq:OUsymVG}. The next step is to build an object of \code{yuima.model} class using the standard constructor \code{setModel}: 

\begin{CodeChunk}
\begin{CodeInput}
#### Definition of an object of yuima.model class ####
 
R> yuima1 <- setModel(drift = "alpha1*(alpha2-X)", jump.coeff = "gamma",
+    jump.variable = "J", solve.variable = c("X"), state.variable = c("X"),         
+    measure.type = "code", measure = list(df = mylaw))
\end{CodeInput}
\end{CodeChunk}

It is worth noticing that the slot \code{measure} of the object \code{yuima1} contains the object \code{mylaw} constructed previously.

\begin{CodeChunk}
\begin{CodeInput} 
R> print(yuima1@measure[[1]])
\end{CodeInput}
\begin{CodeOutput}
An object of class "yuima.law"
Slot "rng":
function(n, eta, t){
  rvg(n, vgC = 0, sigma = sqrt(t), theta = 0, nu = 1/(eta*t))
}

Slot "density":
function(x, eta, t){
  dvg(x, vgC = 0, sigma = sqrt(t), theta = 0, nu = 1/(eta*t))
}

Slot "cdf":
function(q,...){NULL}
<environment: 0x000001e202715cf0>

Slot "quantile":
function(p,...){NULL}
<environment: 0x000001e202715cf0>

Slot "characteristic":
function(u,...){NULL}
<environment: 0x000001e202715cf0>

Slot "param.measure":
[1] "eta"

Slot "time.var":
[1] "t"

Slot "dim":
[1] NA
\end{CodeOutput}
\end{CodeChunk}

We can generate a sample path using the \code{simulate} method in \yuima that we report in Figure \ref{fig:OUVG1}. The simulation scheme in \yuima is based on the Euler discretization, and the small-time increments of the noise $J$ therein are generated by the random number generator stored in \code{mylaw} object.

\begin{CodeChunk}
\begin{CodeInput} 
#### real parameters ####
R> alpha1 <- 0.4; alpha2 <- 0.25; gamma<- 0.25; eta <- 1

#### Sample grid ####
R> n <- 50000 
R> Time <- 1000 
R> sam <- setSampling(Terminal = Time, n = n) 

#### Simulation ####
R> yuima2 <- setYuima(model = yuima1, sampling = sam)
R> true <- list(alpha1 = alpha1, alpha2 = alpha2, gamma = gamma, eta = eta)
R> set.seed(123)
R> yuima3 <- simulate(yuima2, true.parameter = true, sampling = sam)

#### plot sample path ####
R> plot(yuima3)
\end{CodeInput}
\begin{figure}[!h]
\centering 
\includegraphics[width=10cm,height=7cm]{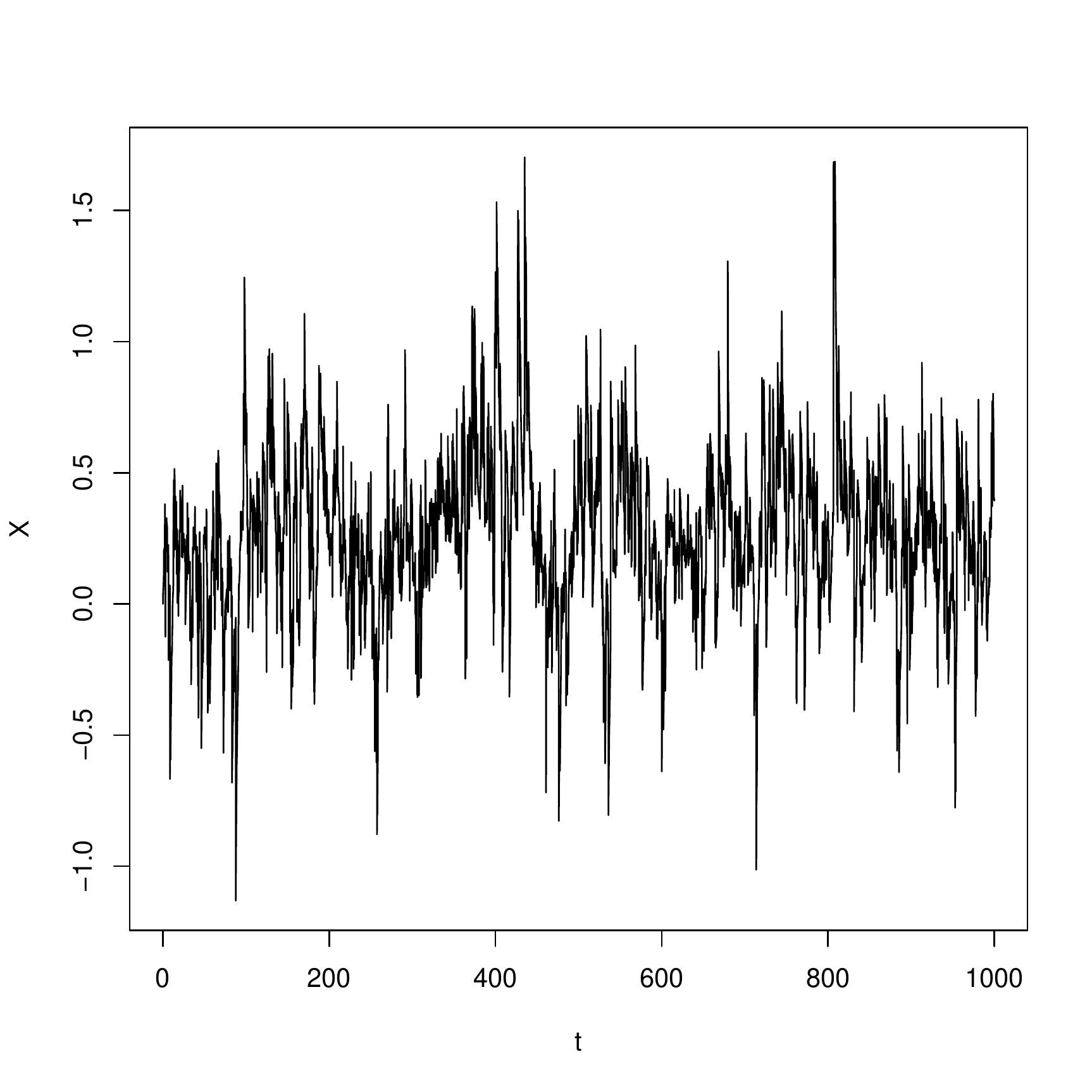}
\caption[Simulated Trajectory of an OU-VG model with parameters $\alpha_1 = 0.4, \alpha_2 = 0.25, \gamma = 0.25, \eta = 1$.]{Simulated Trajectory of an OU-VG model with parameters $\alpha_1 = 0.4, \alpha_2 = 0.25, \gamma = 0.25, \eta = 1$.\label{fig:OUVG1}}
\end{figure}

\end{CodeChunk}

To assess numerically the effectiveness of the three-step estimation procedure discussed in Section \ref{yu:notation and assumption} and Section \ref{yu:main} we re-estimate the model in \eqref{eq:OUsymVG} using the data stored in the object \code{yuima3}. 

\begin{CodeChunk}
\begin{CodeInput} 
#### starting point ####
R> set.seed(123)
R> start <- list(alpha1 = runif(1, 0.01, 2), alpha2 = runif(1, 0.01, 2),
+    gamma = runif(1, 0.01, 2), eta = runif(1, 0.5, 1.5))

#### upper and lower bounds ####
R> upper <- list(alpha1 = 2, alpha2 = 2, gamma = 2, eta = 1.5) 
R> lower <- list(alpha1 = 0.01, alpha2 = 0.01, gamma = 0.01, eta = .5) 

#### GQMLE procedure ####
R> res.VG <- qmleLevy(yuima3, start = start, lower = lower, 
+    upper = upper, Est.Incr = "IncrPar", aggregation = TRUE,
+    joint = FALSE) 
\end{CodeInput}
\end{CodeChunk}

The function \code{qmleLevy} returns an object of \code{yuima.qmleLevy.incr} class that extends the standard class \code{yuima.qmle}. 

\begin{CodeChunk}
\begin{CodeInput} 
R> class(res.VG)
\end{CodeInput}
\begin{CodeOutput}
[1] "yuima.qmleLevy.incr"
attr(,"package")
[1] "yuima"
\end{CodeOutput}
\begin{CodeInput} 
R> slotNames(res.VG)
\end{CodeInput}
\begin{CodeOutput} 
 [1] "Incr.Lev"      "logL.Incr"     "minusloglLevy" "Levydetails"   "Data"  "model"        
 [7] "call"          "coef"          "fullcoef"      "fixed"         "vcov"  "min"          
[13] "details"       "minuslogl"     "nobs"          "method"
\end{CodeOutput}
\end{CodeChunk}

The slot \code{Incr.Lev} is filled with an object of \code{yuima.data} class that contains the estimated unit-time L\'evy increments.  

\begin{CodeChunk}
\begin{CodeInput} 
R> str(res.VG@Incr.Lev, 2)
\end{CodeInput}
\begin{CodeOutput}
Formal class 'yuima.data' [package "yuima"] with 2 slots
  ..@ original.data:'zooreg' series from 1 to 1000
  Data: num [1:1000, 1] 0.141 0.249 -1.219 1.336 0.106 ...
  .. ..- attr(*, "dimnames")=List of 2
  Index:  num [1:1000] 1 2 3 4 5 6 7 8 9 10 ...
  Frequency: 1 
  ..@ zoo.data     :List of 1
\end{CodeOutput}
\end{CodeChunk}

Figure \ref{fig:EstIncrFromOUVG1} reports the trajectory of the estimated unit-time L\'evy increments.

\begin{CodeChunk}
\begin{CodeInput} 
#### Visualization of the estimated unit-time increments ####
R> plot(res.VG@Incr.Lev)
\end{CodeInput}
\begin{figure}[!h]
\centering 
\includegraphics[width=10cm,height=7cm]{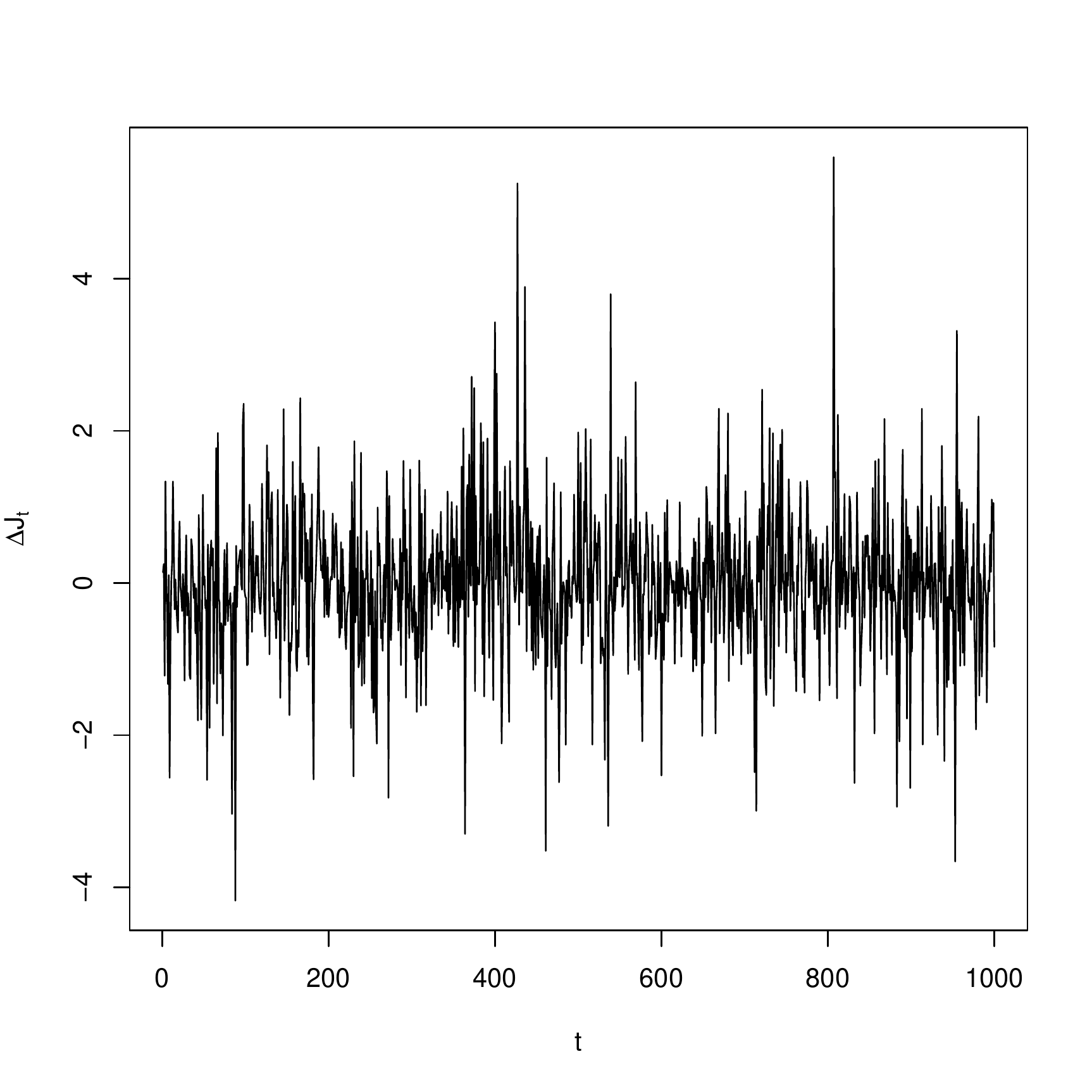} 
\caption[Estimated unit-time increments from the OU-VG model defined in \eqref{eq:OUsymVG}.]{Estimated unit-time increments from the OU-VG model defined in \eqref{eq:OUsymVG}.\label{fig:EstIncrFromOUVG1}}
\end{figure}
\end{CodeChunk}

\subsection{Multivariate L\'evy SDE model}
\label{hm:sec_multivariate.SDE.sim}

In this section, we simulate and estimate a bivariate SDE model driven by two independent symmetric Variance Gamma processes. As done in the previous section we construct the random number generator and the joint density function of the underlying bivariate  L\'evy process using the function developed in the \texttt{VarianceGamma} package. We report below the code for simulating and estimating the process $X_t:=\left[X_{1,t},X_{2,t}\right]^\top$ that satisfies the following system of SDEs:
\begin{equation}
\label{eq:ExampBivariate}
\left.
\begin{array}{l}
d X_{1,t} = \alpha_{1,1}\left(\alpha_{1,2}-X_{1,t}-0.2X_{2,t}\right) d t + \gamma_1 d J_{1,t}\\
d X_{2,t} = \alpha_{2,1}\left(\alpha_{2,2}-X_{2,t}\right) d t + \gamma_2 d J_{2,t}\\
\end{array}
\right.
\end{equation}
where $J_{t}=\left[J_{t,1},J_{t,2}\right]^{\top}$ is a bivariate L\'evy process where the components are two independent symmetric Variance Gamma processes.\newline
The first step is to construct an object of \code{yuima.law} class that contains a random number generator and the joint density of the bivariate L\'evy process $J_{t}$. As done for the model in \eqref{eq:OUsymVG} we use the functions available in the \proglang{R} package \code{VarianceGamma}. The random number generator of the increments can be defined using the following command lines:
\begin{CodeChunk}
\begin{CodeInput}
#### Construction of a bivariate rng function ####
R> myrng2 <- function(n, eta1, eta2, t){
+    res0 <- rvg(n, vgC = 0, sigma = sqrt(t), theta = 0, nu = 1 / (eta1 * t))
+    cbind(res0, rvg(n, vgC = 0, sigma = sqrt(t), theta = 0, nu = 1 / (eta2 * t)))
+  }
\end{CodeInput}
\end{CodeChunk}
Compared with the random number generator used in the univariate case, the result of the function \code{rng} is a two-column matrix where each column contains increments generated from a symmetric Variance Gamma random variable. Exploiting the independence assumption we construct the joint density of the process $J_t$ as a product of two univariate symmetric Variance Gamma densities using the following \proglang{R} function:
\begin{CodeChunk}
\begin{CodeInput}
#### Construction of the joint density ####
R> mydens2 <- function(x, eta1, eta2, t){
+    dvg(x[,1], vgC = 0, sigma = sqrt(t), theta = 0, nu = 1/(eta1 * t)) *
+      dvg(x[,2], vgC = 0, sigma = sqrt(t), theta = 0, nu = 1/(eta2 * t))
+  } 
\end{CodeInput}
\end{CodeChunk}
Using the constructor \code{setLaw}, we build an object of \code{yuima.law} that contains information for simulating the noise $J_t$ and for estimating the parameters in \eqref{eq:ExampBivariate}.
\begin{CodeChunk}
\begin{CodeInput}
R> mylaw2 <- setLaw(rng = myrng2, density = mydens2, dim = 2)
\end{CodeInput}
\end{CodeChunk}
We simulate a trajectory of the model in \eqref{eq:ExampBivariate} using the standard syntax in \yuima as follows:
\begin{CodeChunk}
\begin{CodeInput}
#### Model Definition ####
R> yuima2 <- setModel(drift = c("alpha11*(alpha12-X1-0.2*X2)","alpha21*(alpha22-X2)"), 
+    jump.coeff = matrix(c("gamma1", 0, 0, ""gamma2"), 2, 2), solve.variable = c("X1", "X2"), 
+    state.variable = c("X1", "X2"), measure.type = c("code", "code"), 
+    measure = list(df = mylaw2), jump.variable = "J")

#### Choosing model parameters ####
R> alpha11 = 0.4; alpha12 = 0.25; gamma1 = 0.2; eta1 = 1
R> alpha21 = 0.3; alpha22 = 0.3; gamma2 = 0.1; eta2 = 1
R> true2 <- list(alpha11 = alpha11, alpha12 = alpha12, gamma1 = gamma1, eta1 = eta1,
+    alpha21 = alpha21, alpha22 = alpha22, gamma2 = gamma2, eta2 = eta2)

#### Setting the sample grid ####
R> n2 <- 50000 
R> Time2 <- 1000 
R> sam2 <- setSampling(Terminal = Time2, n = n2)

#### Simulation ####
R> yuima2 <- setYuima(model = yuima2, sampling = sam2) 
R> set.seed(123)
R> yuima2 <- simulate(yuima2, true.parameter = true2, sampling = sam2) 
\end{CodeInput}
\end{CodeChunk}
Figure \ref{fig:SimulatedBivariateProcess} reports the simulated trajectory of each member in the process $X_t:=\left[X_{t,1},X_{t_2}\right]^\top$
\begin{CodeChunk}
\begin{CodeInput}
R> plot(yuima2)
\end{CodeInput}
\begin{figure}[!h]
\centering 
\includegraphics[width=10cm,height=7cm]{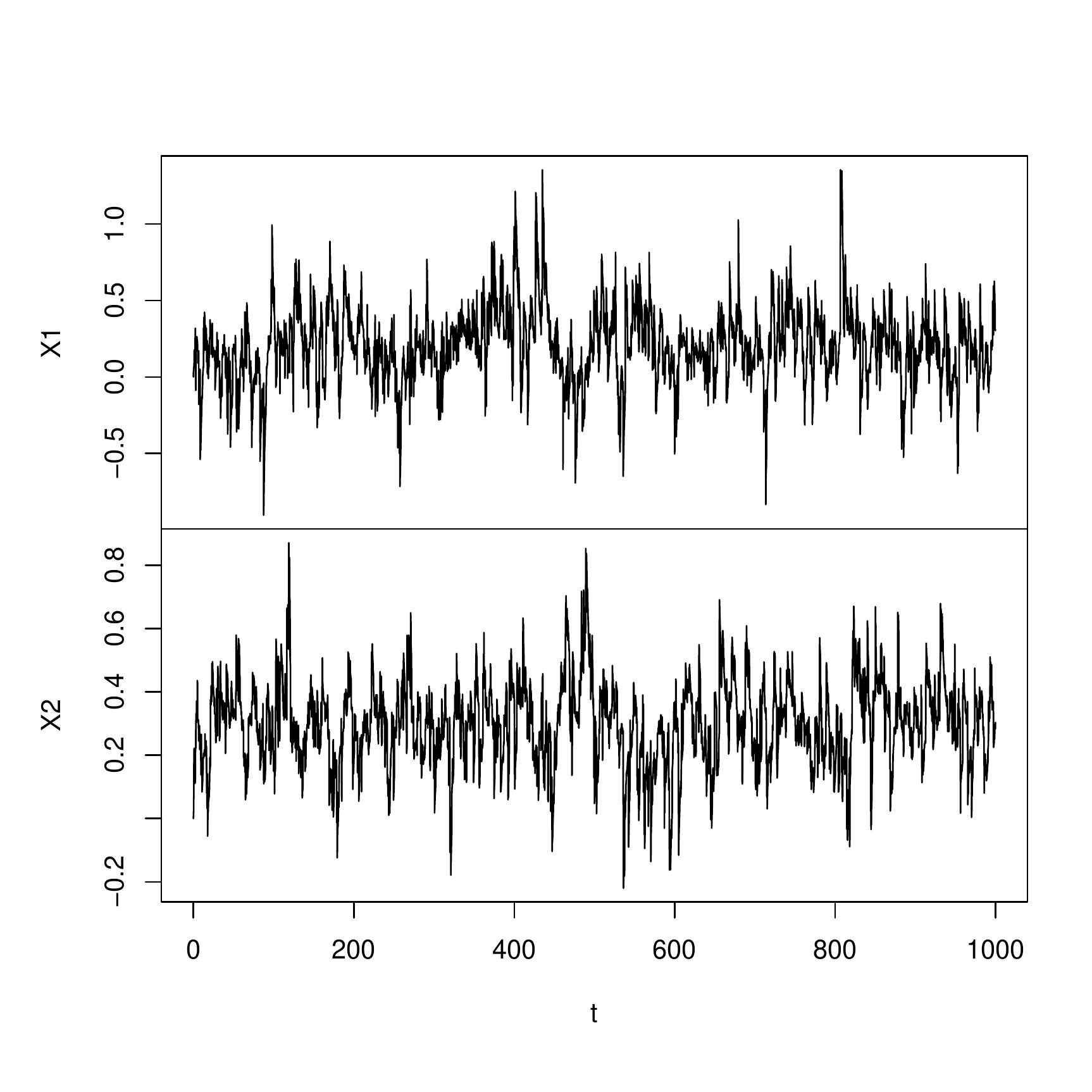} 
\caption[Simulated trajectory of the bivariate process $Y_t$ defined in \eqref{eq:ExampBivariate}.]{Simulated trajectory of the bivariate process $Y_t$ defined in \eqref{eq:ExampBivariate}.\label{fig:SimulatedBivariateProcess}}
\end{figure}
\end{CodeChunk} 

Now we execute the three-step estimation procedure using, as a dataset, the simulated trajectory stored in the slot \code{data}  of the object \code{yuima2}. As done in the univariate case, we select randomly a starting point and we fix upper and lower bounds for each parameter. We select the inputs of the function \code{qmleLevy} to get an object of \code{yuima.qmleLevy.incr} class that contains the estimated unit-time increments of the noise and the L\'evy measure parameters of the bivariate symmetric Variance Gamma process $Z_t$.
\begin{CodeChunk}
\begin{CodeInput}
#### Starting point generation #### 
R> set.seed(123)
R> start2 <- list(alpha11 = runif(1, 0.01, 2), alpha12 = runif(1, 0.01, 2),
+    gamma1 = runif(1, 0.01, 2), eta1 = runif(1, 0.5, 2), alpha21 = runif(1, 0.01, 2), 
+    alpha22 = runif(1, 0.01, 2), gamma2 = runif(1, 0.01, 2), eta2 = runif(1, 0.5, 2))

#### Upper and lower bounds ####
R> upper2 <- list(alpha11 = 2, alpha12 = 2, gamma1 = 2, eta1 = 2,
+    alpha21 = 2, alpha22 = 2, gamma2 = 2, eta2 = 2) 
R> lower2 <- list(alpha11 = 0.01, alpha12 = 0.01, gamma1 = 0.01, eta1 = .5,
+    alpha21 = 0.01, alpha22 = 0.01, gamma2 = 0.01, eta2 = .5) ## set lower bound

#### Estimation ####
R> res.VG2 <- qmleLevy(yuima2, start = start2, lower = lower2, upper = upper2, 
+    Est.Incr = "IncrPar", aggregation = TRUE, joint = FALSE) 
\end{CodeInput}
\end{CodeChunk} 
With the following command lines, we compare the initial values for the optimization routine, the fixed and estimated parameters.
\begin{CodeChunk}
\begin{CodeInput}
#### Starting values ####
unlist(start2)[names(coef(res.VG2))]
\end{CodeInput}
\begin{CodeOutput}
  alpha11   alpha12   alpha21   alpha22    gamma1    gamma2      eta1      eta2 
0.5822793 1.5787272 1.8815299 0.1006574 0.8238641 1.0609299 1.8245261 1.8386286 
\end{CodeOutput}
\begin{CodeInput}
#### Real parameters ####
unlist(true2)[names(coef(res.VG2))]
\end{CodeInput}
\begin{CodeOutput}
alpha11 alpha12 alpha21 alpha22  gamma1  gamma2    eta1    eta2 
   0.40    0.25    0.30    0.30    0.20    0.10    1.00    1.00 
\end{CodeOutput}
\begin{CodeInput}
#### Estimated parameters ####
coef(res.VG2)
\end{CodeInput}
\begin{CodeOutput}
  alpha11   alpha12   alpha21   alpha22    gamma1    gamma2      eta1      eta2 
0.3668973 0.2633624 0.2984875 0.3009404 0.2053790 0.1017661 0.9883838 0.9745951 
\end{CodeOutput}
\end{CodeChunk}
The estimated parameters seem to be precise. The Euclidean norm of the difference between \code{true2} and \code{coef(res.VG2}) is approximately 0.0457 while, applying the same distance between \code{true2} and \code{start2}, it results to be 2.652 with a reduction of $98\%$. We show the standard errors applying the function \code{summary}. 
\begin{CodeChunk}
\begin{CodeInput}
#### Summary ####
summary(res.VG2)
\end{CodeInput}
\begin{CodeOutput}
summary(res.VG2)
Quasi-Maximum likelihood estimation

Call:
qmleLevy(yuima = yuima2, start = start2, lower = lower2, upper = upper2, 
    joint = FALSE, Est.Incr = "IncrPar", aggregation = TRUE)

Coefficients:
          Estimate  Std. Error
alpha11  0.3668973 0.005725729
alpha12  0.2633624 0.002610983
alpha21  0.2984875 0.028302434
alpha22  0.3009404 0.017663594
gamma1   0.2053790 0.023555954
gamma2   0.1017661 0.010761713
eta1     0.9883838 0.098752509
eta2     0.9745951 0.094324105

-2 log L: -494221.9 -494557.8 5335.13 
\end{CodeOutput}
\end{CodeChunk}
The estimated time-unit increments of the bivariate L\'evy noise are available in the slot \code{Incr.Lev} 
\begin{CodeChunk}
\begin{CodeInput}
R> summary(res.VG2@Incr.Lev)
\end{CodeInput}
\begin{CodeOutput}
   Length1    Length2      Class       Mode 
      1000       1000 yuima.data         S4 
\end{CodeOutput}
\begin{CodeInput}
R> plot(res.VG2@Incr.Lev, ylab = c(expression(paste(Delta, J[1, t])),
+    expression(paste(Delta, J[2, t]))), xlab = "t")
\end{CodeInput}
\begin{figure}[!h]
\centering 
\includegraphics[width=10cm,height=7cm]{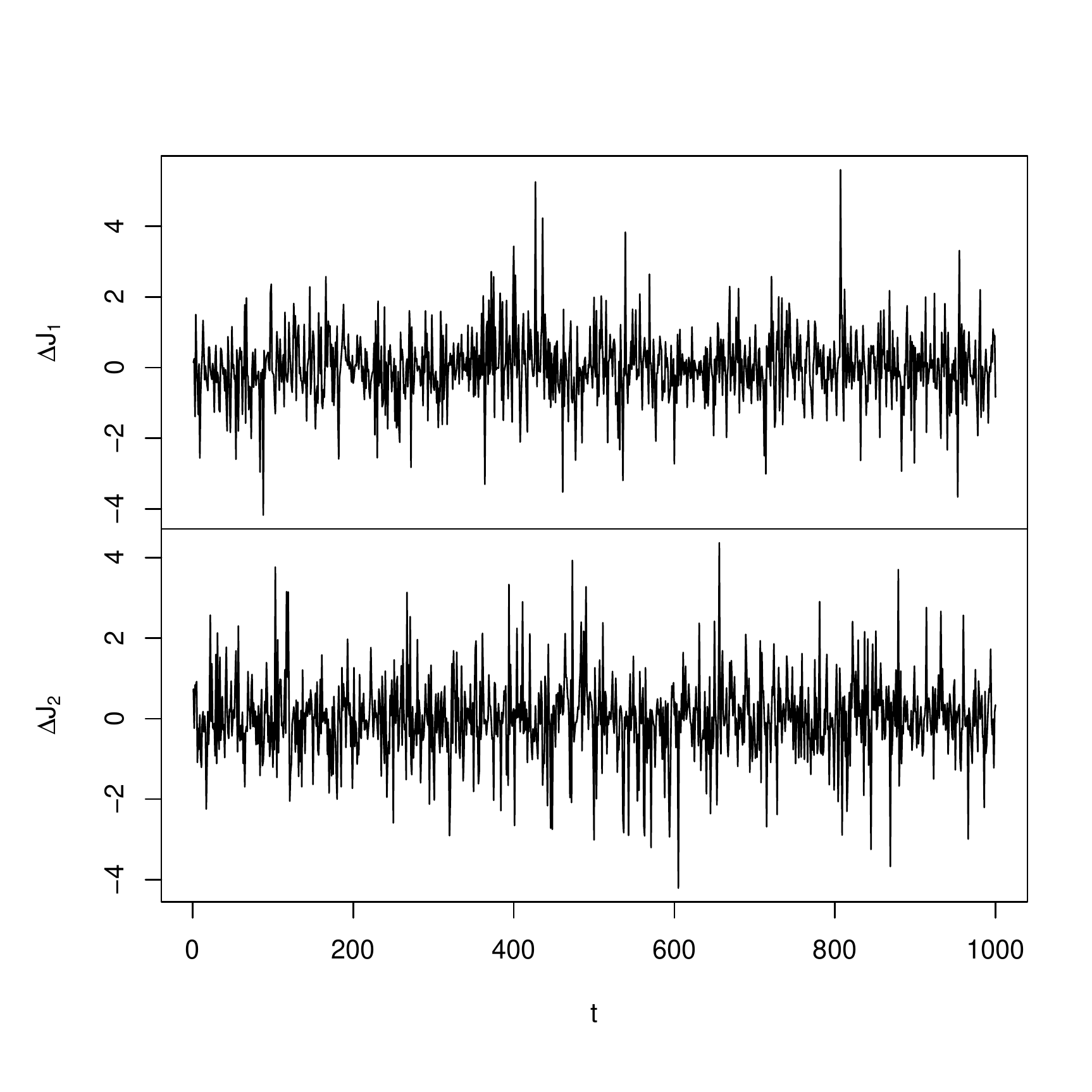}
\caption[Estimated unit-time increments of the bivariate process $J_t$.]{Estimated unit-time increments of the bivariate process $J_t$. \label{fig:BivariateUnitTimeIncr}}
\end{figure}
\end{CodeChunk}

\subsection{Real Data}
\label{hm:sec_real.data}

In this section, we discuss how to estimate a stochastic differential equation driven by a L\'evy process using real data. Once the increments have been obtained, we show how to use them in the two different situations: noise selection and forecasting.
We start with an example that shows how to combine the information stored in an object of \code{yuima.qmleLevy.incr} class with available \proglang{R} packages for selecting a L\'evy measure. The data is  downloaded from \code{yahoo.finance} using the \proglang{R} package \texttt{quantmod} that downloads time-series in an \code{xts} format. We get the closing log-prices of the S\&P500 index ranging from 04 January 1951 to 04 January 2021 using the following command lines:
\begin{CodeChunk}
\begin{CodeInput}
#### Download Dataset ####

R> library(quantmod)
R> getSymbols(Symbols = "^GSPC", from = "1951-01-04", to = "2021-01-04")
R> logprice <- log(GSPC$GSPC.Close)
R> plot(logprice, main = "Closed log-prices of Standard & Poor 500", main.cex = 0.8)
\end{CodeInput}
\begin{figure}[!h]
\centering 
\includegraphics[width=10cm,height=7cm]{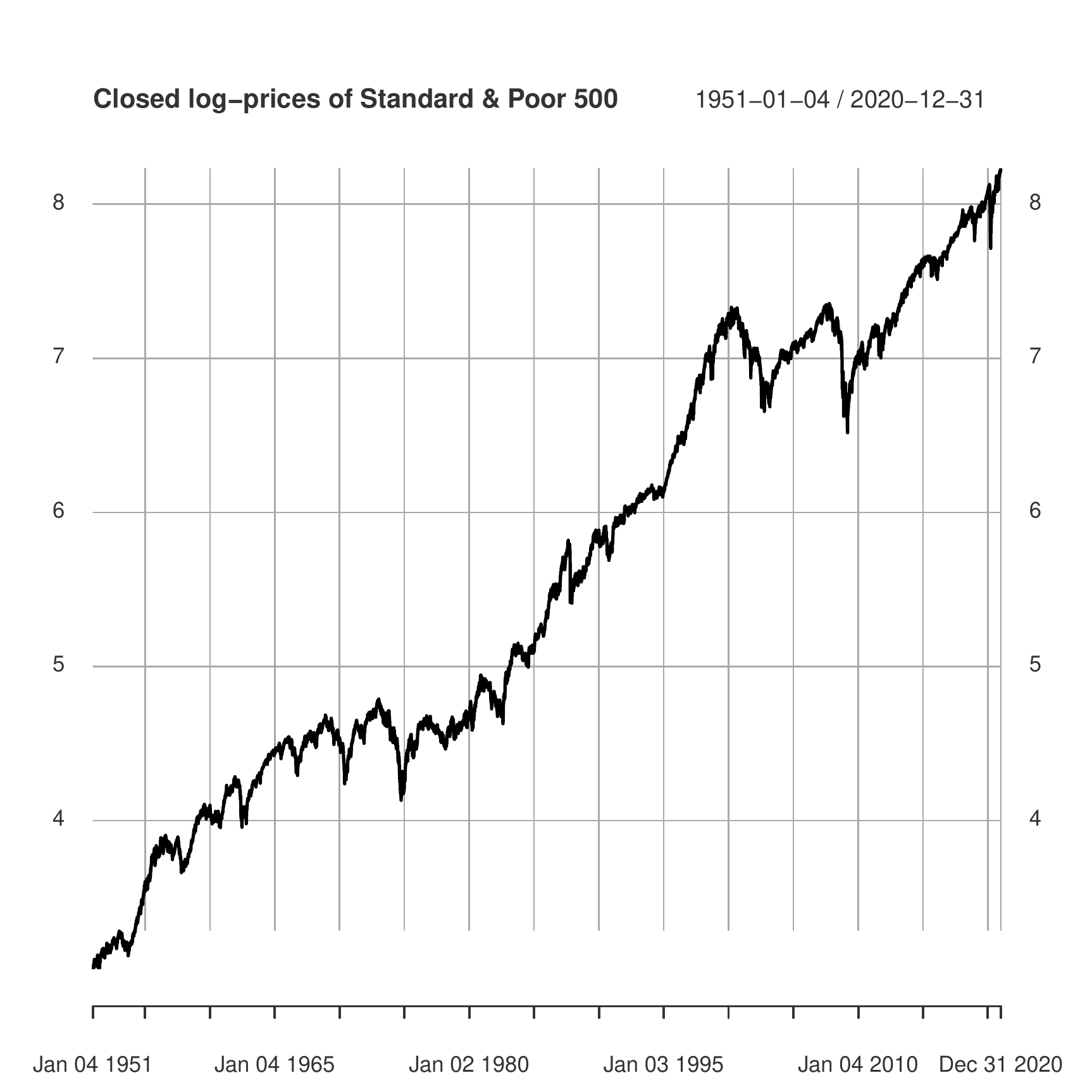}
\caption[Closing log-prices of S\&P500 ranging from 04 January 1951 to 04 January 2021]{Closing log-prices of S\&P500 ranging from 04 January 1951 to 04 January 2021 \label{SP500logprice}}
\end{figure}
\end{CodeChunk}
Figure \ref{SP500logprice} reports the time series used in our example. 
We describe the log-price by the following SDE:
\begin{equation}
d X_t = \left(\alpha_1+\alpha_2 X_t\right)d t + \gamma_1 X_t^{\gamma_2}d J_t,
\label{candidate}
\end{equation}
where $\alpha_1>0$, $\alpha_2 <0$,  $\gamma_1 > 0$ and $\gamma_2 \geq 0$. 
We construct an object of \code{yuima} class that contains the mathematical description of the SDE in \eqref{candidate} and the data. 
\begin{CodeChunk}
\begin{CodeInput}
#### Law Definition ####
R> mylaw3 <- setLaw(dim = 1)

R> #### Model and Data ####
yuima3 <- setModel(drift = "alpha1+alpha2*X", jump.coeff = matrix(c("gamma1*X^gamma2")),
+    measure.type = "code", measure = list(df = mylaw3), jump.variable = "J", 
+    solve.variable = c("X"), state.variable = c("X"))
R> Data <- setData(logprice, delta = 1/30)
R> yuima3 <- setYuima(data = Data, model = yuima3)
R> print(Data)
\end{CodeInput}
\begin{CodeOutput}
Number of original time series: 1
length = 17615, time range [1951-01-04 ; 2020-12-31]

Number of zoo time series: 1
           length time.min time.max      delta
GSPC.Close  17615        0  587.133 0.03333333
\end{CodeOutput}
\end{CodeChunk}
From the structure of the object \texttt{Data}, we observe that the time is expressed in a monthly basis. Therefore, setting $t_0=0$, we have  $T_n=587.133$ and $h=1/30$. 

It is worth noting that the object \code{mylaw} does not require a formal specification for the random number generator and for the density function as done in the previous examples. Indeed, we do not assume any specific form of the L\'evy measure of the process $J$ and the estimation of the unit-time increments described in Section \ref{yu:main} is completely model-free.
\begin{CodeChunk}
\begin{CodeInput}
#### Estimation of time-unit increments ####
R> set.seed(123)
R> start3 <- list(alpha1 = runif(1, min = 10^(-10), max = 1), 
+    alpha2 = runif(1, min = -1, max = -10^(-10)), gamma1 = runif(1, min = 10^(-10), max = 1), 
+    gamma2 = runif(1, min = 0, max = 2))
R> lower3 <- list(alpha1 = 10^(-10), alpha2 = -1, gamma1 = 10^(-10), gamma2 =  0)
R> upper3 <- list(alpha1 = 1, alpha2 =  1, gamma1 = 1, gamma2 =  2)
R> res3 <- qmleLevy(yuima3, start = start3, lower = lower3, upper = upper3, 
+    Est.Incr = "Incr", aggregation = TRUE, joint = FALSE) 
R> summary(res3)
\end{CodeInput}
\begin{CodeOutput}
Quasi-Maximum likelihood estimation

Call:
qmleLevy(yuima = yuima3, start = start3, lower = lower3, upper = upper3, 
    joint = FALSE, Est.Incr = "Incr", aggregation = TRUE)

Coefficients:
           Estimate   Std. Error
gamma1  0.016267709 0.0007539279
gamma2  0.694168099 0.0715834862
alpha1  0.012156634 0.0081051421
alpha2 -0.000635202 0.0015396235

-2 log L: -113838.7 -113858.5  
\end{CodeOutput}
\end{CodeChunk}
Applying \code{logLik} method to \code{res3}, we determine the value for the stepwise GQL function $\mbbh_{1,n}(\hat{\gamma}_1, \hat{\gamma}_2)$ and the value of GQMLE with the following command lines:

\begin{CodeChunk}
\begin{CodeInput}
R> T_n <- tail(index(Data@zoo.data[[1]]),1L)
R> H_1 <- -1/T_n*logLik(res3)[1]
R> GQMLE <-logLik(res3)[2]
R> print(c(H_1, GQMLE))
\end{CodeInput}
\begin{CodeChunk}
[1] -96.94454 56929.23019
\end{CodeChunk}
\end{CodeChunk}

The unit-time increments are stored in the slot \code{res3@Incr.Lev} and they can be extrapolated using the following command lines:
\begin{CodeChunk}
\begin{CodeInput}
#### Time-unit increments ####
R> UnitaryIncr <- as.numeric(res3@Incr.Lev@original.data)
R> plot(UnitaryIncr, ylab = expression(Delta*J[1]), xlab = " ",
+    main = "Estimated Time-Unit Increments", cex.main = 0.8)
\end{CodeInput}
\begin{figure}[!h]
\centering 
\includegraphics[width=10cm,height=7cm]{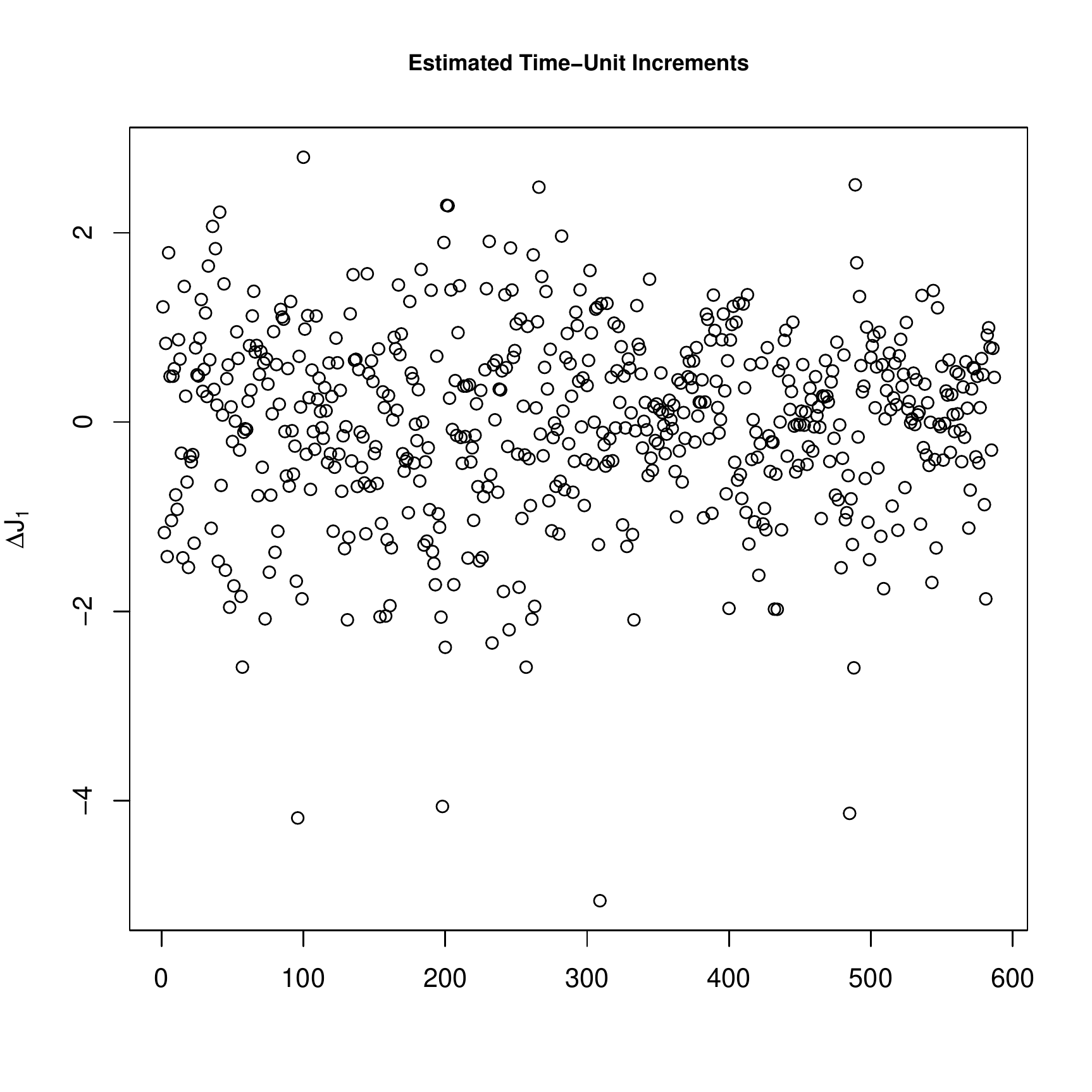}
\caption[Estimated time-unit increments from the real dataset]{Estimated time-unit increments from the real dataset \label{FigIncrReal}}
\end{figure}
\end{CodeChunk}
Figure \ref{FigIncrReal} reports the  estimated time-unit increments of the process $J$. Due to the fact that the object  \code{UnitaryIncr} belongs to the \code{numeric} class, we can apply any method available in \proglang{R} for any \code{numeric} object.
Just for an illustration, in the following command lines, we show how to get the kernel density estimate based on the estimated increments, and compare it with the empirical histogram; a graphical  comparison is reported in Figure \ref{KernelDens}.
\begin{CodeChunk}
\begin{CodeInput}
#### Plot kernel density ####
R> hist(UnitaryIncr, freq = F, nclass = 50, main = "Density of Time-Unit Increments", 
+    cex.main =0.8, xlab = expression(Delta*J[1]), ylab = " ")
R> lines(density(UnitaryIncr),  col = "red")
\end{CodeInput}
\begin{figure}[!h]
\centering
\includegraphics[width=10cm,height=7cm]{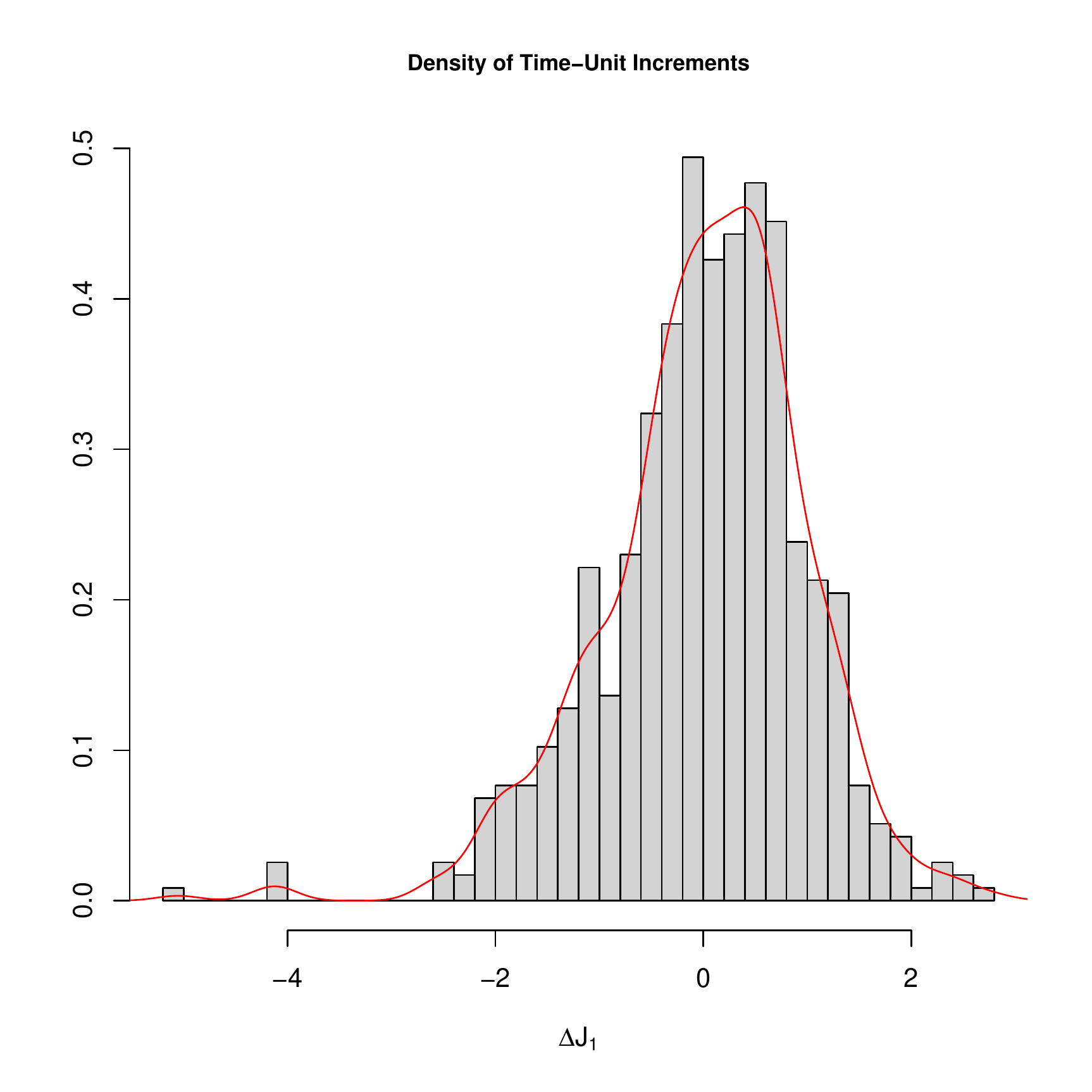}
\caption[Comparison betwen empirical and kernel densities of the time-unit increments]{Comparison betwen empirical and kernel densities of the time-unit increments \label{KernelDens}}
\end{figure}
\end{CodeChunk}
A model selection exercise can be done using the function \code{stepAIC.ghyp} available in the package \texttt{ghyp} that allows the user to compare a list of distributions widely applied in finance.
In particular, based on the Akaike Information Criterion, this function identifies the best model between the Generalized Hyperbolic, the Hyperbolic, the Variance Gamma, the Normal Inverse Gaussian, the Student-$t$ and the Normal distribution.
\begin{CodeChunk}
\begin{CodeInput}
#### Model selection  ####
R> library(ghyp)
R> Comparison <- stepAIC.ghyp(UnitaryIncr)
R> Comparison$best.model
\end{CodeInput}
\begin{CodeOutput}
Asymmetric Hyperbolic Distribution:

Parameters:
 alpha.bar         mu      sigma      gamma 
 1.8398103  0.5383745  0.9075450 -0.5337697 

log-likelihood:
-793.7689

Call:
stepAIC.ghyp(data = UnitaryIncr)
\end{CodeOutput}
\end{CodeChunk}
In our example, the function \code{stepAIC.ghyp} selects the Asymmetric Hyperbolic distribution as the best fitting model with the $\left(\lambda, \bar{\alpha},\mu, \sigma, \gamma\right)$-parametrization for a generic Generalized Hyperbolic distribution. The latter is a normal variance mean mixture with a Generalized Inverse Gaussian subordinator $\tau_{t}$ and, as described in the package documentation \cite{ghyp}, the $\left(\lambda, \bar{\alpha},\mu, \sigma, \gamma\right)$-parametrization requires the characteristic function of $\tau_1$ to be:
\[
\phi_{\tau_1}\left(u\right) = \left(\frac{\varphi}{\varphi-2iu}\right)^{\frac{\lambda}{2}}\frac{K_{\lambda}\left(\sqrt{\xi\left(\varphi-2iu\right)}\right)}{K_{\lambda}\left(\sqrt{\xi\varphi}\right)}
\]
where $\varphi =\bar{\alpha}\frac{K_{\lambda+1}\left(\sqrt{\bar{\alpha}}\right)}{K_{\lambda}\left(\sqrt{\bar{\alpha}}\right)}$, $\xi=\bar{\alpha}\frac{K_{\lambda}\left(\sqrt{\bar{\alpha}}\right)}{K_{\lambda+1}\left(\sqrt{\bar{\alpha}}\right)}$ and $\bar{\alpha}>0$. The real parameters $\mu$ and $\gamma$ control the position and the skewness while $\sigma\geq0$ is a scale parameter for the Generalized Hyperbolic distribution. The Hyperbolic distribution is obtained by setting $\lambda=1$.

In the second example, we show how to use an object of \code{yuima.law} class to generate a new trajectory using the increments of the process $J$. In this case, we need to estimate the increments associated with the interval of length $\Delta t$ (small model-time length). In this example, we use a shorter dataset composed of three years of observations of the  S\&P500 index ranging from 04 January 2018 to 04 January 2021. For the estimation of the increments, the chunk code is exactly the same used in the previous example with only one difference. Indeed, to obtain the $\Delta t$ increments we set the input \code{aggregation} as \code{FALSE}.

\begin{CodeChunk}
\begin{CodeInput}
#### Download dataset ####
R> getSymbols(Symbols = "^GSPC", from = "2018-01-04", to = "2021-01-04")
R> logprice2 <- log(GSPC$GSPC.Close)

#### Model and data ####
R> mylaw4 <- setLaw(dim = 1)
R> yuima4 <- setModel(drift = "alpha1+alpha2*X", jump.coeff = matrix(c("gamma1*X^gamma2")),
+    measure.type = "code", measure = list(df = mylaw4), jump.variable = "J", 
+    solve.variable = c("X"), state.variable = c("X"))
R> Data <- setData(logprice, delta = 1/30)
R> yuima4 <- setYuima(data = Data, model = yuima4)

#### Estimation delta t increments ####
R> set.seed(123)
R> start4 <- list(alpha1 = runif(1, min = 10^(-10), max = 5), 
+    alpha2 = runif(1, min = -1, max = -10^(-10)), gamma1 = runif(1, min = 10^(-10), max = 1), 
+    gamma2 = runif(1, min = 0, max = 2))
R> lower4 <- list(alpha1 = 10^(-10), alpha2 =  -1, gamma1 = 10^(-10), gamma2 =  0)
R> upper4 <- list(alpha1 = 5, alpha2 =  -10^(-10), gamma1 = 1, gamma2 =  2)
R> res4 <- qmleLevy(yuima4, start = start4, lower = lower4, upper = upper4,
+    Est.Incr = "Incr", aggregation = FALSE, joint = FALSE)
R> summary(res4)
\end{CodeInput}
\begin{CodeOutput}
Quasi-Maximum likelihood estimation

Call:
qmleLevy(yuima = yuima4, start = start4, lower = lower4, upper = upper4, 
    joint = FALSE, Est.Incr = "Incr", aggregation = FALSE)

Coefficients:
          Estimate  Std. Error
gamma1  0.08111166 0.005677789
gamma2  0.00000000 0.065495083
alpha1  2.01137182 1.367347271
alpha2 -0.25019063 0.171158938

-2 log L: -4207.318 -4210.077 
\end{CodeOutput}
\end{CodeChunk}  

Using the estimated $\Delta t$ increments in the slot \code{res4@Incr.Lev} we can build an object of \code{yuima.law} class that internally uses a random number generator that samples from the data in \code{res4@Incr.Lev}.
\begin{CodeChunk}
\begin{CodeInput}
#### yuima.law Definition ####
R> mydata <- as.numeric(res4@Incr.Lev@original.data)
R> myrndEmp <- function(n, mydata){
+    sample(mydata, size = n)
+  }
R> mylaw5 <- setLaw(rng = myrndEmp)
\end{CodeInput}
\end{CodeChunk} 

The object \code{mylaw} contains a random number generator that uses the \proglang{R} function \code{sample},  however, the user can apply more advanced sampling methods from packages available from CRAN. We can simulate one-year trajectory of the S\&P500 log prices in YUIMA. Figure \ref{Forecast} reports the simulated sample path.

\begin{CodeChunk}
\begin{CodeInput}
#### Generation 1 year trajectory ####
R> yuima5 <- setModel(drift = "alpha1+alpha2*X", jump.coeff = matrix(c("gamma1*X^gamma2")),
+    measure.type = "code", measure = list(df = mylaw5), jump.variable = "J",
+    solve.variable = c("X"), state.variable = c("X") 
+    xinit = as.numeric(tail(logprice2, 1L)))
R> samp5 <- setSampling(Initial = 0, Terminal = 24, n = 24*30)
R> yuima5 <- setYuima(model = yuima5, sampling = samp5)
R> true5 <- as.list(coef(res4))
R> true5$mydata <- mydata
R> set.seed(123)
R> yuima5 <- simulate(yuima5, true.parameter = true5, sampling = samp5)
R> plot(yuima5, main = "Forecasted 24 months trajectory of the S&P500 Index", 
+    cex.main = 0.8)  
\end{CodeInput}
\begin{figure}[!h]
\centering
\includegraphics[width=10cm,height=7cm]{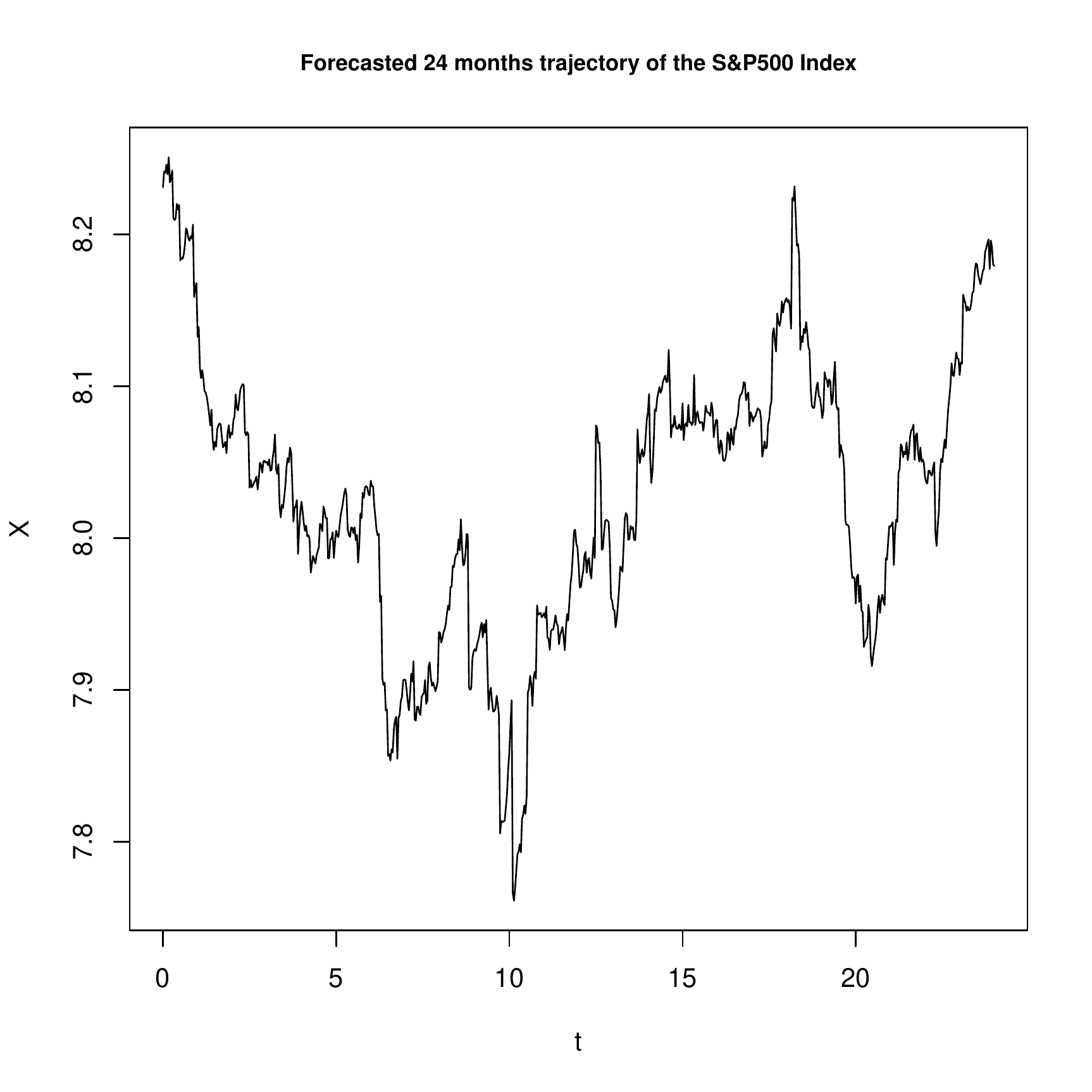}
\caption[24 months simulated trajectory of the S\&P500 log-price series]{24 months simulated trajectory of the S\&P500 log-price series\label{Forecast}}
\end{figure}
\end{CodeChunk}
In Figure \ref{Forecast}, the 24 months trajectory displays an oscillatory behavior. It fluctuates around the long term mean that can be estimated using the ratio $-\frac{\hat{\alpha_1}}{\hat{\alpha_2}}\approx 8.0394$.

\subsection*{Acknowledgement}
We thank the anonymous reviewers for their valuable comments.
This work was partly supported by 
JST CREST Grant Number JPMJCR14D7, Japan.

\bibliographystyle{abbrv}




\end{document}